\documentclass[a4paper,onecolumn,11pt,accepted=2024-04-28]{quantumarticle}
\pdfoutput=1
\usepackage{hyperref}
\hypersetup{colorlinks,linkcolor={blue},citecolor={blue},urlcolor={blue}}
\usepackage{amsfonts,amssymb,amsmath}
\usepackage{mathtools}
\usepackage{dcolumn}
\usepackage{bm}
\usepackage{physics}
\usepackage[numbers,sort&compress]{natbib}

\newcommand{\nn}{\nonumber\\}
\newcommand{\beq}{\begin{equation}}
\newcommand{\eeq}{\end{equation}}

\newcommand{\bigO}[1]{\mathcal{O}\left( #1 \right)}
\newtheorem{theorem}{Theorem}[section]

\newcommand{\eq}[1]{(\ref{eq:#1})}
\newcommand{\fig}[1]{\hyperref[fig:#1]{Figure~\ref*{fig:#1}}}
\renewcommand{\sec}[1]{\hyperref[sec:#1]{Section~\ref*{sec:#1}}}
\newcommand{\timesteps}{r}
\newcommand{\smalltime}{M}
\newcommand{\totalrows}{R}
\newcommand{\solrat}{\mathcal{R}}

\newcommand{\dervbnd}{\mathcal{D}}

\newenvironment{proof}{\noindent \textbf{{Proof.~} }}{\hfill $\square$}

\allowdisplaybreaks

\begin{document}

\title{Template demonstrating the quantumarticle document class}

\title{Quantum algorithm for time-dependent differential equations using Dyson series}
\author{Dominic W. Berry and Pedro C. S. Costa}

\affiliation{School of Mathematical and Physical Sciences, Macquarie University, Sydney, New South Wales 2109, Australia}

\begin{abstract}
Time-dependent linear differential equations are a common type of problem that needs to be solved in classical physics. Here we provide a quantum algorithm for solving time-dependent linear differential equations with logarithmic dependence of the complexity on the error and derivative. As usual, there is an exponential improvement over classical approaches in the scaling of the complexity with the dimension, with the caveat that the solution is encoded in the amplitudes of a quantum state. Our method is to encode the Dyson series in a system of linear equations, then solve via the optimal quantum linear equation solver. Our method also provides a simplified approach in the case of time-independent differential equations.
\end{abstract}

\maketitle

Quantum computers offer tremendous potential advantages in computing speed, but it is a long-standing challenge to find speedups to important computational tasks.
A major advance for quantum algorithms was the development of a way of solving systems of linear equations by Harrow, Hassidim, and Lloyd \cite{Harrow_2009}.
This algorithm provides an exponential improvement in complexity in terms of the dimension as compared to classical solution, with some caveats.
In particular, the matrix needs to be given in an oracular way (rather than as classical data), and the solution is encoded in amplitudes of a quantum state.

There has therefore been a great deal of follow-up work on applications of solving linear equations where the result is potentially useful despite these limitations.
For example, in discretised partial differential equations (PDEs), the discretisation yields a large set of simultaneous equations.
Then one may aim to obtain some global feature of the solution which may be obtained by sampling from the prepared quantum state.
This was the principle used in \cite{Clader13}, and there have been considerable further advances since then \cite{Ashley2016,Childs2021}.

There is also the question of how to solve an \emph{ordinary} differential equation (ODE).
That is, spatial dimensions are not given explicitly (though the set of equations may have been obtained by a spatial discretization of a PDE), and the task is to solve the time evolution.
The original algorithm for this task used linear multistep methods \cite{BerryJPA14}, and a further improvement was to use a Taylor series encoded in a larger system of linear equations to obtain complexity logarithmic in the allowable precision \cite{BerryCMP17}.
For this and other quantum algorithms for solving differential equations, it is required that the equations are dissipative in order to provide an efficient solution.
Krovi \cite{Krovi} showed how to solve the time-independent case with an alternative dissipative condition based on the logarithmic norm rather than the eigenvalues of the matrix having non-positive real part as in \cite{BerryCMP17}.

That then leaves open the question of how to solve time-dependent differential equations.
Algorithms for time-dependent Hamiltonian evolution have been given based on a Dyson series in \cite{Kieferova2019,Low2018}.
A method for time-dependent differential equations was given in Ref.~\cite{Childs2020}.
That used a spectral method to provide near-linear dependence on $T$, with poly-logarithmic factors in $T$ and the allowable error.
The poly-logarithmic factors were not explicitly given in \cite{Childs2020}, but they are quite large.
In particular, $n$ is chosen logarithmically in Eq.~(8.6) of that work, then there is a factor of $n^4$ for the complexity in Eq.~(8.15) of Ref.~\cite{Childs2020}.
Another limitation of the approach of Ref.~\cite{Childs2020} is that it requires smoothness of the differential equation, with a complexity depending on arbitrary-order derivatives of the parameters.
Reference \cite{Fang} used a time-marching strategy for time-dependent differential equations, but gave quadratic dependence of the complexity on the overall evolution time $T$, making it considerably more costly.
Moreover, it only analyses the homogeneous case.

Here we provide an algorithm for inhomogeneous time-dependent differential equations with significantly improved logarithmic dependence on the error and linear dependence on $T$ up to logarithmic factors.
Our logarithmic factors are significantly improved over those in \cite{Childs2020}.
Moreover, we need no smoothness condition as in Ref.~\cite{Childs2020}.
Our oracle complexity is entirely independent of derivatives of the parameters, and the complexity in terms of elementary gates depends only on the log of the first derivatives of parameters.
We combine methods from both \cite{BerryCMP17} and \cite{Kieferova2019,Low2018}.
We use a block matrix similar to \cite{BerryCMP17}, but we do \emph{not} use extra lines of the block matrix to implement terms in the series as in that work.
Instead we construct this matrix via a block encoding using a Dyson series, in an analogous way as the block encodings in \cite{Kieferova2019,Low2018}.

In \sec{form} we describe the form of the solution in terms of the sum of a solution to the homogeneous equation and a particular solution.
Then in \sec{cond} we express the method in terms of a set of linear equations and determine the condition number.
In \sec{block} we describe the method to block encode the linear equations.
Lastly we use these results to give the overall complexity in \sec{final}, and conclude in \sec{conc}.

\section{Form of solution}
\label{sec:form}
In this section we summarise standard methods of solving a linear differential equation via a Dyson series.
A general ordinary linear differential equation has form
\begin{equation}
\bm{\dot x}(t) = A(t)\bm{x}(t)+\bm{b}(t), \label{eq:general0}
\end{equation}
where $\bm{b}(t) \in \mathbb{C}^{N}$ is a vector function of $t$, $A(t)\in \mathbb{C}^{N\times N}$ is a coefficient matrix and 
$\bm{x}(t)\in \mathbb{C}^{N}$ is the solution vector as a function of variable $t$. We are also given an initial condition $\bm{x}(t_0) = \bm{x}_0$. Note that we can always express a higher-order ODE as a system of first-order differential equations by defining new parameters.

We will write the general solution to \eq{general0} as
\begin{equation}
\bm{x}(t) = \bm{x}_H(t) + \bm{x}_P(t),
\label{eq:gen_sol}
\end{equation}
where $\bm{x}_H(t)$ is the solution to the homogeneous equation $\dot{\bm{x}}_H = A(t)\bm{x}_H$ and $\bm{x}_P(t)$ is a particular solution. In the next two subsections we show how to compute them.

\subsection{Solution to the homogeneous equation}
First we solve a differential equation of the form
\begin{equation}
\dot{\bm{x}}_H = A(t)\bm{x}_H.
\end{equation}
A general solution for $\bm{x}_H$ can be expressed in the form of a Dyson series $\bm{x}_H(t)=W(t,t_0) \bm{x}_H(t_0)$ with
\begin{align}
 W(t,t_0) &\coloneqq \sum_{k=0}^{\infty} \frac{1}{k!} \int_{t_0}^t dt_1 \int_{t_0}^t dt_2 \cdots \int_{t_0}^t dt_k \, \mathcal{T} A(t_1) A(t_2) \cdots A(t_k) ,   
\end{align}
where $\mathcal{T}$ indicates time ordering.
We will briefly review the algorithm for the first segment. We can explicitly order the solution and truncate the infinite sum at a cutoff $K$ to give
\begin{align}
\label{eq:trunc}
W_K(t,t_0) &\coloneqq \sum_{k=0}^{K} \int_{t_0}^t dt_1 \int_{t_0}^{t_1} dt_2 \cdots  \int_{t_0}^{t_{k-1}} dt_k\, A(t_1) A(t_2) \cdots A(t_k).
\end{align}
Note that the $k!$ disappears because the times without ordering have $k!$ permutations, so sorting them gives $k!$ multiplicity.
This expression is just equivalent to the Dyson series solution for Hamiltonian evolution used in Refs.~\cite{Kieferova2019,Low2018}, and the fact that $A$ is a more general matrix does not affect the form of the equations.
Then, using $A_{\max}=\max_t \|A(t)\|$,
\begin{align}
\left\| W_K(t,t_0) - W(t,t_0) \right\| &\le \sum_{k=K+1}^\infty \frac {(A_{\max}\Delta t)^k}{k!} \nn
&\le \frac{(A_{\max}\Delta t)^{K+1}}{(K+1)!}\exp(A_{\max}\Delta t) \nn
&= \bigO{ \frac{(A_{\max}\Delta t)^{K+1}}{(K+1)!} } ,   
\end{align}
with $\Delta t=t-t_0$, provided $\Delta t$ is chosen so $A_{\max}\Delta t$ is $\bigO{1}$.
Here and throughout the paper we are taking $\| \cdot \|$ to be the spectral norm for operators, or the 2-norm for vectors.

If the goal is to give error no larger than $\epsilon$ due to the truncation, then we can choose $K\propto\log(1/\epsilon)/\log\log(1/\epsilon)$.
To solve for a longer time we break the time into $r$ segments and use the truncated Dyson series on each segment.
Therefore the error in each segment should be no larger than $\epsilon/r$.
The choice of $K$ should then be
\begin{equation}\label{eq:choiceK}
    K = \bigO{\frac{\log(r/\epsilon)}{\log\log(r/\epsilon)}} .
\end{equation}

\subsection{Particular solution}
Generalising to an inhomogeneous equation
\beq
\dot{\bm{x}} = A(t)\bm{x}+\bm{b}(t) \, ,
\eeq
the particular solution (i.e.\ with initial condition of zero) is of the form $\bm{x}_P(t)=\bm{v}(t,t_0)$ 
with
\begin{align}
\label{eq:part_sol}
\bm{v}(t,t_0) &\coloneqq \sum_{k=1}^{\infty} \int_{t_0}^t dt_1 \int_{t_0}^{t_1} dt_2 \cdots \int_{t_0}^{t_{k-1}} dt_{k} A(t_1) A(t_2) \cdots A(t_{k-1})\bm{b}(t_{k}) \, .
\end{align}
Note that this is very similar to $W(t,t_0)$ for the homogeneous solution, except we have replaced $A(t_k)$ with $\bm{b}(t_{k})$.
Note also that, in $W(t,t_0)$ we have $k=0$ giving a product of \emph{none} of the $A$, so that term gives the identity.
Here the sum starts from $k=1$, and $k=1$ gives the integral of $\bm{b}(t_{1})$, so we can rewrite the solution as
\begin{align}
\bm{v}(t,t_0) &= \sum_{k=2}^{\infty} \int_{t_0}^t dt_1 \int_{t_0}^{t_1} dt_2 \cdots \int_{t_0}^{t_{k-1}} dt_{k} \, A(t_1)  A(t_2) \cdots A(t_{k-1})\bm{b}(t_{k}) + \int_{t_0}^t dt_1 \, \bm{b}(t_{1})\, .
\end{align}
To see that this is the correct form of the solution, one can use
\begin{align}
    \dot{\bm{x}}_P(t) &= \dot{\bm{v}}(t,t_0) \nonumber \\*
    &= \sum_{k=2}^{\infty} \int_{t_0}^{t} dt_2 \cdots \int_{t_0}^{t_{k-1}} dt_{k} \, A(t)  A(t_2) \cdots A(t_{k-1})\bm{b}(t_{k}) + \bm{b}(t)\nn
    &= A(t) \sum_{k=1}^{\infty} \int_{t_0}^{t} dt_1 \int_{t_0}^{t_1} dt_2 \cdots \int_{t_0}^{t_{k-1}} dt_{k}\, A(t_1) \cdots A(t_{k-1})\bm{b}(t_{k}) + \bm{b}(t)\nn
    &= A(t)\bm{x}_P + \bm{b}(t) \, ,
\end{align}
which is the form of the differential equation.
Truncating the solution, we have
\begin{align}
\bm{v}_K(t,t_0) &\coloneqq \sum_{k=1}^{K} \int_{t_0}^t dt_1 \int_{t_0}^{t_1} dt_2 \cdots \int_{t_0}^{t_{k-1}} dt_{k}\, A(t_1) A(t_2) \cdots A(t_{k-1})\bm{b}(t_{k}) \, .
\end{align}
The complete approximate solution, as the sum of the homogeneous and particular solutions with truncations, is then
\begin{equation}\label{eq:solstep}
  \bm{x}_K(t) = W_K(t,t_0)\bm{x}(t_0)+\bm{v}_K(t,t_0)  .
\end{equation}

Defining $b_{\max} = \max_t \| \bm{b}(t)\|$,
then the error due the truncation of the Dyson series is
\begin{align}
\| \bm{x}_K(t) - \bm{x}(t) \| &\le \left\| W_K(t,t_0) - W(t,t_0) \right\| \!\times\! \| \bm{x}(t_0) \|\nn 
&\quad + \left\|\sum_{k=K+1}^{\infty} \int_{t_0}^t dt_1 \int_{t_0}^{t_1} dt_2 \cdots \int_{t_0}^{t_{k-1}} dt_{k}\,A(t_1) A(t_2) \cdots A(t_{k-1})\bm{b}(t_{k})\right\| \nn
&\le \bigO{ \frac{(A_{\max}\Delta t)^{K+1}}{(K+1)!} \| \bm{x}(t_0) \|}  + \sum_{k=K+1}^{\infty} \frac{A_{\max}^{k-1}\Delta t^k b_{\max}}{k!} \nn
&=  \mathcal{O}\left( \frac{(A_{\max}\Delta t)^{K+1}}{(K+1)!} \| \bm{x}(t_0)\| + \frac{A_{\max}^K \Delta t^{K+1} b_{\max}}{(K+1)!} \right),
\label{eq:errorbnd}
\end{align}
where the second norm in the first line results from $\norm{\bm{v}(t,t_0)- \bm{v}_K(t,t_0)}$. Again we require that $A_{\max}\Delta t$ is $\bigO{1}$.
In our solution we will discretise the integrals, and to bound the overall error we also bound the error due to the discretisation; that is discussed in \sec{final}. 

This complete expression in Eq.~\eqref{eq:errorbnd} has an extra factor of $\| \bm{x}(t_0)\|$ because we are considering the error in the state.
The error in $W$ being no greater than $\epsilon$ implies that the error in the state is no larger than $\epsilon$ times the norm of $\bm{x}$, which is the form we require the bound on the error for our theorem.
Appropriately bounding the error in the second term in Eq.~\eqref{eq:errorbnd} can require taking $K$ somewhat larger than in Eq.~\eqref{eq:choiceK}.
The choice of $K$ is discussed further in \sec{final}.

\subsection{The time-independent solution}

We are also going to look at the case where $A \in \mathbb{C}^{N\times N}$ is time-independent, which gives our usual ODE to solve
\begin{equation}
    \bm{\dot x}(t) = A\bm{x}(t)+\bm{b}. \label{eq:general}
\end{equation}
In \cite{BerryCMP17} this was solved by using the Taylor series.
Here we apply a similar principle, except that we will encode into the block matrix in a simpler way.
Here we summarise the Taylor series solution, then give the encoding into the matrix in the next section.

The solution can easily be seen by substituting a constant $A$ into the solution for the time-varying $A$.
We obtain
\beq
\label{eq:TruncTay}
W_K(t,t_0) = \sum_{k=0}^{K} \frac {(A\Delta t)^k}{k!} ,
\eeq
and
\beq
\bm{v}_K(t,t_0) = \sum_{k=1}^{K} \frac{(A\Delta t)^{k-1}}{k!}\bm{b} \, .
\eeq
The error in the solution can be bounded as
\begin{align}
\| \bm{x}_K(t) - \bm{x}(t) \| &=  \mathcal{O} \left(\frac{(\|A\|\Delta t)^{K+1}}{(K+1)!} \| \bm{x}(t_0)\|+ \frac{\|A\|^K \Delta t^{K+1} \|\bm{b}\|}{(K+1)!} \right).
\end{align}

\section{Encoding in linear equations}
\label{sec:cond}

Next we describe how the Dyson series (or Taylor series) for the solution may be encoded into block matrices which can then be solved via the quantum algorithm for solving linear systems.
The preceding section described how the series may be applied to accurately approximate the solution over a short time.
In the usual way, for solving the solution over a longer time we divide the total time $T$ into $\timesteps$ steps of length $\Delta t = T/\timesteps$.
The solution for the shorter times is obtained by a series, and all time steps are encoded in a block matrix.

\subsection{Time-dependent equations}

\subsubsection{Encoding}

First we describe the more complicated case of the block encoding of the Dyson series solution.
This step shares many similarities with techniques for time-independent equations~\cite{BerryCMP17} but elements of the block matrix $\mathcal{A}^{-1}$ will need to be computed through integration described in the preceding section.
The initial time for the multiple steps will be taken to be zero without loss of generality, since it is always possible to apply a time shift in the definitions.
Then we can use notation $t_0$ for the starting time for individual steps of the Dyson series.

To illustrate the method, we first give the example of three time steps, where the encoding is
\beq
\label{eq:examp3}
\begin{bmatrix}
    \openone & 0 & 0 & 0 & 0 & 0 &  0 \\
    -V_1 & \openone & 0 & 0 & 0 & 0 &  0 \\
    0 & -V_2 & \openone & 0 & 0 & 0 & 0 \\
    0 & 0 & -V_3 & \openone & 0 & 0 & 0 \\
    0 & 0 & 0 & -\openone & \openone & 0 & 0 \\
    0 & 0 & 0 & 0 & -\openone & \openone & 0 \\
    0 & 0 & 0 & 0 & 0 & -\openone & \openone
\end{bmatrix}
\begin{bmatrix}
\bm{x}(0) \\ \tilde{\bm{x}}(\Delta t) \\ \tilde{\bm{x}}(2\Delta t) \\ \tilde{\bm{x}}(3\Delta t) \\ \tilde{\bm{x}}(3\Delta t) \\ \tilde{\bm{x}}(3\Delta t) \\ \tilde{\bm{x}}(3\Delta t)
\end{bmatrix}
=
\begin{bmatrix}
\bm{x}(0) \\ \bm{v}_1 \\ \bm{v}_2 \\ \bm{v}_3 \\ \bm{0} \\ \bm{0} \\ \bm{0}
\end{bmatrix} 
\eeq
where
\begin{align}
    \bm{v}_m &\coloneqq \bm{v}_K(m\Delta t,(m-1)\Delta t), \\
    V_m &\coloneqq W_K(m\Delta t,(m-1)\Delta t),
\end{align}
and $\tilde{\bm{x}}$ indicates the approximate solution for $\bm{x}$. 
The top row (numbered as $0$ here) sets the initial condition and 
rows 1 to 3 give steps of the solution for time $\Delta t$ as described by Eq.~\eq{solstep}.
These rows all describe forward evolution as 
\begin{equation}
    \tilde{\bm{x}}(m\Delta t) = V_m \bm{x}((m-1)\Delta t) + \bm{v}_m
\end{equation}
for $m\in \{1,2,3 \}$. In rows $4$ to $6$, the system is constant in time, $ \tilde{\bm{x}}(m\Delta t) = \bm{x}((m-1)\Delta t)$. Including rows after the final evolution steps might appear unnecessary, but these additional rows will boost the success probability of the resulting quantum algorithm~\cite{berry2014high}. 
Note that $\bm{v}_m$ and $V_m$ involve multiple integrals over times, which would need to be addressed via the block encoding which is described in \sec{block}.
This is a major departure from the method for time-independent differential equations from \cite{BerryCMP17}, where it was possible to encode the terms of the sum via extra lines in the block matrix.

In general, we can define the system of linear equations as
\beq
\mathcal{A} \mathcal{X} = \mathcal{B},
\eeq
where  $\mathcal{A}$ is a block matrix,  $\mathcal{A} \in \mathbb{C}^{N\totalrows \times N\totalrows}$ and $\mathcal{X}$ and $\mathcal{B}$ are vectors $\mathcal{X},\mathcal{B} \in \mathbb{C}^{N\totalrows} $. 
Here $\totalrows$ is the number of time steps, including those steps at the end where the solution is held constant.
Taking $\totalrows-\timesteps \propto \timesteps$ will give a success probability roughly corresponding to the square of the amplitude of the solution. 
If the solution has not decayed too much, then this success probability will give an acceptable overhead to the complexity.

The individual blocks in $\mathcal{A},\mathcal{X},\mathcal{B}$ can be given explicitly as
\begin{align}
\mathcal{B}_m &= \begin{cases}
\bm{x}(0), &  m=0 \\
\bm{v}_m, & 0<m\le \timesteps \\
0, & m > \timesteps
\end{cases} \\
\mathcal{X}_m &= \begin{cases}
\bm{x}(0), &  m=0 \\
\tilde{\bm{x}}(m\Delta t), & 0< m\le \timesteps \\
\tilde{\bm{x}}(\timesteps\Delta t), & m > \timesteps
\end{cases} \\
\mathcal{A}_{mn} &= \begin{cases}
\openone , & m=n\\
-V_n, & (m=n+1) \wedge (n\le \timesteps) \\
-\openone, & (m=n+1) \wedge (n>\timesteps) \\
0, & {\rm otherwise.}
\end{cases}\label{eq:Adef}
\end{align}
It is easily checked that this is a general form of the example given above.
By construction, this a lower bidiagonal matrix which will be be useful for computing its properties.

\subsubsection{Condition number}

The complexity of solving the system of linear equations is proportional to the condition number of $\mathcal{A}$.
To bound the condition number of $\mathcal{A}$, we just need to bound the norm of $\mathcal{A}$ and its inverse.
The norm of $\mathcal{A}$ is easily bounded as $\mathcal{O}(1)$, but the norm of the inverse will depend on how much the approximate solution $\tilde{\bm{x}}$ can grow as compared to the initial condition $\bm{x}(0)$ and driving $\bm{v}_m$.

This means that the condition number can be well-bounded provided the differential equations are stable.
That can be obtained provided the real parts of the eigenvalues of $A(t)$ are non-positive.
This was the approach applied in \cite{BerryCMP17}, where a diagonalisation was used so the complexity was also dependent on the condition number of the matrix that diagonalises $A$.
It has been pointed out in \cite{Krovi} that the complexity of this approach can be bounded in a number of other ways which do not depend on the diagonalisability of $A$.

As discussed in \cite{Krovi}, a useful way to describe the stability of the differential equation is in terms of the logarithmic norm of $A(t)$, which can be given as the maximum eigenvalue of $[A(t)+A^\dagger(t)]/2$.
A standard property of the logarithmic norm $\mu(A(t))$ is that
\begin{equation}
    \|W(t,t_0)\| \le \exp\left( \int_{t_0}^t \mu(A(t))\right) .
\end{equation}
This means that, if the logarithmic norm is non-positive, then $\|W(t,t_0)\| \le 1$.
That is, this approach can be used to bound the norm of the time-ordered exponential, in a similar way as the exponential is bounded in the time-independent case in \cite{Krovi}.
In the following we will require that the logarithmic norm is non-positive for stability, though other conditions bounding the norm of the time-ordered exponential can be used.

Starting with  $\mathcal{A}$ from Eq.~\eq{Adef}, it can be separated into a sum of the identity (the main diagonal) and a matrix which is just the $-V_n$ and $\openone$ on the offdiagonal.
Using the triangle inequality, the norm can therefore be upper bounded as
\beq
\norm{\mathcal{A}} \le  1+ \max_m \|V_m\|.
\eeq
Here $\max_m \|V_m\|$ is obtained because the spectral norm of the block-diagonal matrix can be given as the maximum of the spectral norms of the blocks.
The maximum is over the values of $m$ from $1$ to $M$.
Now, $V_m$ is defined in terms of $W_K$ which is an approximation of the exact evolution over time $\Delta t$.
We will require that the overall solution is given to accuracy $\epsilon<1$, so the norm of $W_K$ cannot deviate by more than $\epsilon$ from the norm for $W$, which is upper bounded by $1$ according to our stability requirement on the logarithmic norm.
That means $\max_m \|V_m\|=\mathcal{O}(1)$, and so $\|\mathcal{A}\| = \bigO{1}$.
An alternative bound that does not depend on the stability is
\begin{align}
\|V_m\| &=  \norm{ W_K(m\, \Delta t,(m-1)\Delta t) } \nn
&\le \norm{\sum_{k=0}^{K} \frac{(A_{\max}\Delta t)^k}{k!}} \nn
&< \norm{\exp(A_{\max}\Delta t)}.
\end{align}
We choose the number of steps such that $A_{\max}\Delta t=\bigO{1}$, which again gives the upper bound $\|\mathcal{A}\| = \bigO{1}$.

Next, we bound the norm of the inverse.
Since $\mathcal{A}$ is a lower bidiagonal matrix, it has the explicit form of the inverse~\cite{kilicc2013inverse}
\beq
\label{eq:inverA_def}
(\mathcal{A}^{-1})_{mn} = \begin{cases}
\openone , & m=n\\
\prod_{\ell=n}^{m-1} V_{\ell}, & (m>n) \wedge (m\le \timesteps+1) \wedge (n\le \timesteps) \\
\prod_{\ell=n}^{\timesteps} V_{\ell}, & (m>n) \wedge (m> \timesteps+1) \wedge (n\le \timesteps) \\
\openone, & (m>n) \wedge (n>\timesteps) \\
0, & {m<n.}
\end{cases}
\eeq
In our example for $\timesteps=3$, $\mathcal{A}^{-1}$ takes the form
\beq
\mathcal{A}^{-1} = \begin{bmatrix}
    \openone & 0 & 0 & 0 & 0 & 0 &  0 \\
    V_1 & \openone & 0 & 0 & 0 & 0 &  0 \\
    V_2 V_1 & V_2 & \openone & 0 & 0 & 0 & 0 \\
    V_3 V_2 V_1 & V_3 V_2 & V_3 & \openone & 0 & 0 & 0 \\
    V_3 V_2 V_1 & V_3 V_2 & V_3 & \openone & \openone & 0 & 0 \\
    V_3 V_2 V_1 & V_3 V_2 & V_3 & \openone & \openone & \openone & 0 \\
    V_3 V_2 V_1 & V_3 V_2 & V_3 & \openone & \openone & \openone & \openone 
\end{bmatrix} .
\eeq

Again we can use the triangle inequality.
We express $\mathcal{A}^{-1}$ as a sum of matrices, each of which contains one of the diagonals.
For each of the block diagonal matrices, the spectral norm is given by the maximum spectral norm of the blocks on the diagonal.
 More explicitly, we expand $\mathcal{A}^{-1}$ in the sum of block-diagonal matrices
\begin{equation}
    \mathcal{A}^{-1} = \sum_{k=0}^{\totalrows-1} \mathcal{A}^{\rm inv}_k
\end{equation}
where
\begin{equation}
    (\mathcal{A}^{\rm inv}_k)_{mn} = \begin{cases}
\prod_{\ell=n}^{\min(m-1,\timesteps)} V_{\ell}, & m=n+k  \\
0, & {m<n,}
\end{cases}
\end{equation}
and we take the convention that a product with no factors gives the identity.
We have by the triangle inequality
\beq
\| \mathcal{A}^{-1} \| \le \sum_{k=0}^{\totalrows-1} \| \mathcal{A}^{\rm inv}_k \|.
\eeq
Then
\begin{equation}
    \| \mathcal{A}^{\rm inv}_k \| \le \max_{n} \left\| \prod_{\ell=n}^{\min(n+k-1,\timesteps)} V_{\ell} \right\| ,
\end{equation}
so
\begin{align}\label{spectral_inverse}
\| \mathcal{A}^{-1} \| &\le \totalrows \max_{k} \| \mathcal{A}^{\rm inv}_k \| \nn
& \le \totalrows \max_{m\ge n} \left\| \prod_{\ell=n}^{m-1} V_{\ell} \right\|.
\end{align}

The bound in~\eqref{spectral_inverse} can be interpreted as $\totalrows$ times the norm of the largest block of $\mathcal{A}^{-1}$.
This will be an identity or some products of $V_\ell$s corresponding to an approximation of the time-evolution operator over a period of time within the interval $[t_0,T]$.
We will choose the parameters such that the approximation is within $\epsilon$ of the exact operator, and so the norm is within $\epsilon$ of the exact operator.
According to our condition on the stability in terms of the logarithmic norm, the time-ordered exponential has its spectral norm upper bounded by 1.
The product of $V_\ell$ gives the solution of the homogeneous equation approximated using a Dyson series, and we will choose the parameters such that the error in this solution is no more than $\epsilon$.
This means that the error in the spectral norm of the product is also no more than $\epsilon$, and is upper bounded by $1+\epsilon$.
Hence the spectral norm of the inverse is upper bounded as
\beq
\| \mathcal{A}^{-1} \| \le \totalrows (1 + \epsilon) = \bigO{\totalrows},
\eeq
where we have used the fact that we are considering complexity for small $\epsilon$.
Therefore the condition number of $\mathcal{A}$ is
\beq \label{eq:condno}
\kappa_{\mathcal{A}} = \| \mathcal{A} \| \times \| \mathcal{A}^{-1} \| = \bigO{\totalrows}.
\eeq

\subsection{Time-independent equations}

The encoding of the time-dependent case is exactly the same of the independent case, except that now the $V_m$ are independent of $m$, so we can write for example
\begin{equation}
\label{eq:examp4}
\mathcal{A} = \begin{bmatrix}
\openone & 0 & 0 & 0 & 0 & 0 &  0 \\
-V & \openone & 0 & 0 & 0 & 0 &  0 \\
0 & -V & \openone & 0 & 0 & 0 & 0 \\
0 & 0 & -V & \openone & 0 & 0 & 0 \\
0 & 0 & 0 & -\openone & \openone & 0 & 0 \\
0 & 0 & 0 & 0 & -\openone & \openone & 0 \\
0 & 0 & 0 & 0 & 0 & -\openone & \openone
\end{bmatrix},
\end{equation}
where now we have
\begin{equation}\label{eq:Tayser}
    V = W_K(m\Delta t,(m-1)\Delta t) = \sum_{k=0}^{K} \frac {(A\Delta t)^k}{k!} ,
\end{equation}
which is independent of $m$.
Similarly, the $\bm{v}_m$ are replaced with
\beq\label{eq:bser}
\bm{v} = \sum_{k=1}^{K} \frac{A^{k-1}\Delta t^k}{k!}\bm{b} \, .
\eeq
We can write the general form in a similar way as for the time-dependent case, except that the $V_m,\bm{v}_m$ are replaced with the time-independent $V,\bm{v}$.

In exactly the same way as for the time-dependent case, the spectral norm can be upper bounded by expressing the matrix as a sum of block-diagonal matrices and using the triangle inequality 
\begin{equation}
    \| \mathcal{A} \| \leq 1 + \|V\|.
\end{equation}
Again $\|V\|$ will be an accurate approximation of the exponential, which is upper bounded by 1 according to the stability criterion, and so $\|\mathcal{A}\|=\mathcal{O}(1)$.
Similarly to the time-dependent case, the spectral norm of $\mathcal{A}^{-1}$ can be upper bounded by using the triangle inequality on the explicit expression for the inverse.
Again the stability condition guarantees that the norm of the exponential is no larger than 1, and the condition on the approximation of the solution being with error $\epsilon$ implies that the norm of the powers of $V$ is within $\epsilon$ of 1.
Since there is a sum over $\totalrows$ of these norms, we again have $\mathcal{A}^{-1}=\bigO{\totalrows}$ and so
\beq
\kappa_{\mathcal{A}} = \| \mathcal{A} \| \times \| \mathcal{A}^{-1} \| = \bigO{\totalrows}.
\eeq

\section{Block encoding}
\label{sec:block}

\subsection{Time-dependent block encoding}

Here we address how to give a block encoding of the matrix in order to apply the quantum linear equation solver.
The block encoding is not trivial, because the blocks we have given above are composed of multiple integrals.
For our result we just assume that the matrix $A(t)$ is given by a block encoding, rather than considering how it would specifically be done for a particular encoding (such as sparse matrix oracles).
Similarly, we will assume that there is a block encoding for a preparation of the state corresponding to $\bm{b}(t)$ as well as the initial state $\bm{x}_0$.
We assume that the block encoding can be given the time register as a quantum input.
The overall complexity is then in terms of the number of calls to these block encodings.

More specifically, given the target system of dimension $N$, a time register of dimension $N_t$ and ancillas of dimension $N_A,N_b$, there are unitaries $U_A$, $U_b$, and $U_x$ such that
\begin{align}
    _A\!\bra{0} U_A \ket{0}_A\ket{n_t} &= \frac 1{\lambda_A} A(t) \ket{n_t},\\
    _b\!\bra{0} U_b \ket{0}_b\ket{n_t}\ket{0}_s &= \frac 1{\lambda_b} \ket{n_t} \ket{\bm{b}(t)}_s, \\
     U_x \ket{0}_s &= \frac 1{\lambda_x} \ket{\bm{x}_0}_s,
\end{align}
where $n_t$ is a binary encoding of the time $t$, the subscripts $A,b$ on the state indicate the ancillas for the block encodings of $A(t)$ and $\bm{b}(t)$, and the subscript $s$ indicates the target system.
The operator $A(t)$ acts upon the target system, which is omitted in the first line for simplicity.
The quantities $\lambda_A,\lambda_b,\lambda_x$ are the factors for the block encoding and are assumed to be known.
The value of $\lambda_x$ is just that needed to ensure normalisation.
The value of $\lambda_b$ also ensures normalisation, but since the norm of $\bm{b}(t)$ can vary over time we define the oracle in terms of a block encoding so that $\lambda_b$ can be taken to be independent of time.
It would also be possible to define an oracle with a unitary preparation and time-dependent $\lambda_b$, but that would make the later algorithms much more complicated.
We apply the standard encoding of the vector in amplitudes of the state so, for example,
\begin{equation}
    \ket{\bm{x}(t)} = \sum_{n=0}^{N-1} x_n(t) \ket{n},
\end{equation}
for computational basis states $\ket{n}$.
Note that this state is not defined in a normalised way, which is why we have, for example, division by $\lambda_x$ above.
We also use the standard assumption that the block-encoding unitaries can be applied in a controlled way and we can also apply their inverses.

There are standard methods for combining block encodings of matrices we will use.
When we are multiplying operations in a block encoding, the general procedure is that if matrices $A$ and $B$ are block encoded as
\beq
\bra{0} U \ket{0} = \frac{A}{\lambda_A}, \qquad \bra{0} V \ket{0} = \frac{B}{\lambda_B},
\eeq
then we have
\beq
\bra{0}\bra{0}VU \ket{0}\ket{0} = \frac{BA}{\lambda_A\lambda_B}.
\eeq
The understanding is that $U$ and $V$ are acting on different ancillas.
That is, we need the ancilla space for both when multiplying.
This means that if $A$ and $B$ are block encoded with respective $\lambda$-values of $\lambda_A$ and $\lambda_B$ then the $\lambda$-value for the product is $\lambda_A\lambda_B$.

In adding operations, we would build a controlled operation that performs
\beq
C = \ket{0}\bra{0} \otimes U + \ket{1}\bra{1} \otimes V
\eeq
controlled by a qubit ancilla.
If we aim to block encode $aA+bB$ (with positive $a,b$), then we can use a state preparation operation $P$ such that
\begin{equation}
    P \ket{0} = \frac 1{\sqrt{a\lambda_A+b\lambda_B}}\left(\sqrt{a\lambda_A} \ket{0} + \sqrt{b\lambda_B} \ket{1}\right) .
\end{equation}
Then we obtain
\beq
\bra{0}\bra{0} P^\dagger CP\ket{0}\ket{0} = \frac 1{a\lambda_A+b\lambda_B}(aA+bB).
\eeq
The first $P$ puts the qubit ancilla in a superposition, then $C$ applies either $U$ or $V$ controlled on that ancilla, and $P^\dagger$ inverts the preparation.
The resulting $\lambda$-value is then $a\lambda_A+b\lambda_B$, so we are just applying the same arithmetic to determine the $\lambda$-values as for the operators.
If we want negative weights in $a$ or $b$ then the sign can be applied using $C$, and the resulting $\lambda$-value is $|a|\lambda_A+|b|\lambda_B$.

These primitives mean that whenever we have a polynomial in something that is block encoded we can construct a block encoding for the polynomial.
The value of $\lambda$ is determined by using exactly the same arithmetic as in the polynomial for the operators, except we take the absolute values of any weights.
This reasoning also holds for discretised integrals.
The integral can be discretised into a sum, and the value of $\lambda$ for the discretised integral can be determined using the same arithmetic for the $\lambda$-values.
In particular, we use ancilla registers with times to give an approximation of the time integral in the Dyson series, in a similar way as was done for Hamiltonian simulation.

Here we want to block encode an operation of the following form, given for the example of three time steps
\beq
\mathcal{A}=
\begin{bmatrix}
    \openone & 0 & 0 & 0 & 0 & 0 &  0 \\
    -V_1 & \openone & 0 & 0 & 0 & 0 &  0 \\
    0 & -V_2 & \openone & 0 & 0 & 0 & 0 \\
    0 & 0 & -V_3 & \openone & 0 & 0 & 0 \\
    0 & 0 & 0 & -\openone & \openone & 0 & 0 \\
    0 & 0 & 0 & 0 & -\openone & \openone & 0 \\
    0 & 0 & 0 & 0 & 0 & -\openone & \openone
\end{bmatrix} .
\eeq
This matrix can be written in the form
\begin{align}
\label{eq:Asum}
\mathcal{A}&=
    \begin{bmatrix}
    \openone & 0 & 0 & 0 & 0 & 0 &  0 \\
    0 & \openone & 0 & 0 & 0 & 0 &  0 \\
    0 & 0 & \openone & 0 & 0 & 0 & 0 \\
    0 & 0 & 0 & \openone & 0 & 0 & 0 \\
    0 & 0 & 0 & 0 & \openone & 0 & 0 \\
    0 & 0 & 0 & 0 & 0 & \openone & 0 \\
    0 & 0 & 0 & 0 & 0 & 0 & \openone
\end{bmatrix}
-\begin{bmatrix}
0 & 0 & 0 & 0 & 0 & 0 & 0 \\
    \openone & 0 & 0 & 0 & 0 & 0 &  0 \\
    0 & \openone & 0 & 0 & 0 & 0 &  0 \\
    0 & 0 & \openone & 0 & 0 & 0 & 0 \\
    0 & 0 & 0 & \openone & 0 & 0 & 0 \\
    0 & 0 & 0 & 0 & \openone & 0 & 0 \\
    0 & 0 & 0 & 0 & 0 & \openone & 0
\end{bmatrix}
\begin{bmatrix}
    V_1 & 0 & 0 & 0 & 0 & 0 &  0 \\
    0 & V_2 & 0 & 0 & 0 & 0 & 0 \\
    0 & 0 & V_3 & 0 & 0 & 0 & 0 \\
    0 & 0 & 0 & \openone & 0 & 0 & 0 \\
    0 & 0 & 0 & 0 & \openone & 0 & 0 \\
    0 & 0 & 0 & 0 & 0 & \openone & 0 \\
    0 & 0 & 0 & 0 & 0 & 0 & \openone
\end{bmatrix}.
\end{align}
Here there are two operations that are trivial to implement.
One is the identity, and the other is an increment on the time register.
It is not unitary, because the top row is zero.
It can easily be block encoded by incrementing the register in a non-modular way, and using the carry qubit.
That is, recall that the block encoding involves a projection onto the $\ket{0}$ state, so if the carry qubit is flipped to $\ket{1}$ then that part is eliminated in the block encoding.

The difficult operation to block encode is the one with $V_m$ on the diagonal.
To block encode that matrix, we use the intermediate matrix
\beq
\mathcal{A}(\delta t) =
\begin{bmatrix}
A(\delta t) & 0 & 0 & 0 & 0 & 0 &  0 \\
0 & A(\Delta t +\delta t) & 0 & 0 & 0 & 0 & 0 \\
0 & 0 & A(2\Delta t +\delta t) & 0 & 0 & 0 & 0 \\
0 & 0 & 0 & 0 & 0 & 0 & 0 \\
0 & 0 & 0 & 0 & 0 & 0 & 0 \\
0 & 0 & 0 & 0 & 0 & 0 & 0 \\
0 & 0 & 0 & 0 & 0 & 0 & 0
\end{bmatrix},
\eeq
where $\delta t$ is used to index an offset between $0$ and $\Delta t$.
This matrix can be block encoded using the block encoding of $A$ with a time register as input.
This may be achieved simply by using a number of time intervals for each $\Delta t$ that is a power of 2.
Then we may use the qubits encoding $\delta t$, and the qubits encoding the line in the block matrix, together to give the time input for the block encoding of $A(t)$.
For the zeros in the lower-right of $A(\delta t)$ we can flip an ancilla qubit to eliminate that part in the block encoding.
That is most easily performed when $M$ is a power of 2 as well, so a single qubit will flag the lower-right part of the matrix.

Then we use this matrix as input to a truncated Dyson series of the form
\begin{align}
\label{eq:truncDys}
&\sum_{k=0}^{K} \int_{t_0}^{t_0+\Delta t} dt_1 \int_{t_0}^{t_1} dt_2 \cdots \int_{t_0}^{t_{k-1}} dt_k \, \mathcal{A}(t_1) \mathcal{A}(t_2) \cdots \mathcal{A}(t_k) .
\end{align}
In order to block encode the Dyson series, the same procedure as in \cite{Kieferova2019} or \cite{Low2018} can be used.
To be more specific on the complexity, there are $K$ time registers which are generated in equal superposition states and then sorted via a sorting network (as is needed for quantum algorithms).
The sort has an optimal complexity of $\bigO{K\log K}$ steps, each of which has a complexity corresponding to the number of bits used to store $\delta t$.
Using $\smalltime$ for the number of these time steps, the number of bits is $\log \smalltime$, and so the gate complexity is $\bigO{K\log K\log \smalltime}$.
The block encoding of each of the $K$ applications of $\mathcal{A}(\delta t)$ is used, which has complexity of $K$ calls to the block encoding of $A(t)$.

If the inequality testing approach of \cite{Low2018} is used then the complexity of the preparation of the time register is $\bigO{K\log M}$, because there are $K-1$ inequality tests on registers of size $\log M$.
Despite this, in practice it is preferable to use the sorting approach because it provides an improved constant factor in the complexity (which is ignored here in the order scaling).

The preparation of the register with values of $k$ can be performed simply using a unary representation and controlled rotations.
The precision of the truncated Dyson series needs to be $\bigO{\epsilon/\timesteps}$, and since there are $K$ rotations each needs precision $\bigO{\epsilon/(K\timesteps)}$, and therefore has complexity $\bigO{\log(K\timesteps/\epsilon)}$. That complexity is no larger than $\bigO{\log(K)\log(\timesteps/\epsilon)}$, so multiplying by $K$ for the $K$ rotations gives $\bigO{K\log(K)\log(\timesteps/\epsilon)}$.

It is also possible to perform the preparation of $k$ using the improved approach of \cite{SuPRXQuantum21} based on inequality testing.
That yields similar complexity for the order scaling because the number of bits needed is $\bigO{K\log K}$.
In the method of \cite{Low2018} for preparing time registers, the preparation over $k$ only needs an equal superposition, which may be prepared with trivial $\bigO{\log K}$ complexity \cite{Sanders2020}.

Choosing $\lambda_A\Delta t\le 1$ (and the sorting approach for time registers), the Dyson series in Eq.~\eq{truncDys} is block encoded with a $\lambda$-value no more than $e$.
In particular, in Ref.~\cite{Kieferova2019} the normalisation factor for the Dyson series sum with times prepared by a sort is given in Eq.~(55) of that work.
That is for a Hermitian Hamiltonian, but the analysis is unchanged for the general matrix.
In the notation of \cite{Kieferova2019}, $\lambda$ is equivalent to $\lambda_A$ here, and $T/r$ is equivalent to $\Delta t$ here.
That means the normalisation factor ($\lambda$-value) is no larger than $e^{\lambda_A\Delta t}\le e$ with $\lambda_A\Delta t\le 1$.

This result can also be seen using the principles described above for determining $\lambda$-values for block-encodings of polynomials.
Equation~\eq{truncDys} can be described by using weights of $1/k!$ then a time-ordering of the values of $t_j$.
Using the same arithmetic for the $\lambda$-values as for the operators, we obtain the value of $\lambda$
\begin{align}
&\sum_{k=0}^{K} \frac 1{k!} \int_{t_0}^{t_0+\Delta t} dt_1 \int_{t_0}^{t_0+\Delta t} dt_2 \cdots \int_{t_0}^{t_0+\Delta t} dt_k \, \lambda_{A}^k = \sum_{k=0}^{K} \frac {(\lambda_{A}\Delta t)^k}{k!} < e^{\lambda_A\Delta t}\le e .
\end{align}

There is a further constant factor in the block encoding of the sum for $\mathcal{A}$ in Eq.~\eq{Asum}, which may be ignored in our analysis because we are providing the order scaling.
That is, in Eq.~\eq{Asum} the matrix with $V_j$ is the Dyson series, which is block encoded with $\lambda$-value no larger than $e$.
The identity has $\lambda$-value of 1, as does the increment operator (the matrix with identities in the off-diagonal).
In the usual way for describing the block encoding of sums and products of operators, the overall $\lambda$-value is no larger than $e+1$.

The linear equations to solve are of the form $\mathcal{A} \mathcal{X} = \mathcal{B}$, or
\beq \label{eq:lineq}
\begin{bmatrix}
    \openone & 0 & 0 & 0 & 0 & 0 &  0 \\
    -V_1 & \openone & 0 & 0 & 0 & 0 &  0 \\
    0 & -V_2 & \openone & 0 & 0 & 0 & 0 \\
    0 & 0 & -V_3 & \openone & 0 & 0 & 0 \\
    0 & 0 & 0 & -\openone & \openone & 0 & 0 \\
    0 & 0 & 0 & 0 & -\openone & \openone & 0 \\
    0 & 0 & 0 & 0 & 0 & -\openone & \openone
\end{bmatrix}
\begin{bmatrix}
\bm{x}(0) \\ \bm{x}(\Delta t) \\ \bm{x}(2\Delta t) \\ \bm{x}(3\Delta t) \\ \bm{x}(3\Delta t) \\ \bm{x}(3\Delta t) \\ \bm{x}(3\Delta t)
\end{bmatrix}
=
\begin{bmatrix}
\bm{x}_0 \\ \bm{v}_1 \\ \bm{v}_2 \\ \bm{v}_3 \\ \bm{0} \\ \bm{0} \\ \bm{0}
\end{bmatrix}.
\eeq
We therefore need to consider the complexity of preparing the vector on the right $\mathcal{B}$, which is composed of $\bm{x}_0$ which we are given an oracle for, as well as $\bm{v}_j$ which needs to be constructed by block encoding integrals of $\bm{b}(t)$.
The integral to encode is of the form
\begin{align}
\label{eq:bDys}
\bm{v}_K(t_0+\Delta t,t_0)
&= \sum_{k=1}^{K} \int_{t_0}^{t_0+\Delta t} dt_1 \int_{t_0}^{t_1} dt_2 \cdots \int_{t_0}^{t_{k-1}} dt_{k}\, A(t_1) A(t_2) \cdots A(t_{k-1})\bm{b}(t_{k}) \, .
\end{align}
The block encoding of this vector can be applied in almost an identical way as for the block encoding of the truncated Dyson series, except the initial block encoding of $A(t_k)$ is replaced with the block encoding of the preparation of $\ket{\bm{b}(t_{k})}$.
Therefore, the gate complexity of preparing and sorting the time registers is again $\bigO{K\log K\log \smalltime}$, or $\bigO{K\log \smalltime}$ if one were simply to perform inequality testing as in \cite{Low2018}.
There are $K-1$ calls to block encodings of $A(t)$, as well as a single call to the block encoding of $\bm{b}(t_{k})$.

For the block encoding of Eq.~\eq{bDys}, the choice $\lambda_A\Delta t\le 1$ means that $\lambda$-value of the block encoding is no larger than $(e-1)\lambda_b \Delta t$.
As before, the $\lambda$-value can be determined using the same arithmetic, and we can use a factor of $1/k!$ followed by a time ordering of $t_j$.
That gives the $\lambda$ value as
\begin{align}
\sum_{k=1}^{K} \frac 1{k!} \int_{t_0}^{t_0+\Delta t} dt_1 \int_{t_0}^{t_0+\Delta t} dt_2 \cdots \int_{t_0}^{t_0+\Delta t} dt_{k} \, \lambda_A^{k-1} \lambda_b 
&= \frac{\lambda_b}{\lambda_A} \sum_{k=1}^{K} \frac{(\lambda_A\Delta t)^k}{k!}\nn
&< \frac{\lambda_b}{\lambda_A} (e^{\lambda_A\Delta t}-1)\nn 
&\le (e-1)\lambda_b \Delta t \, ,
\end{align}
for $\lambda_A\Delta t\le 1$.

For the complete preparation, it is necessary to prepare $\bm{x}(0)$ as well.
We would initially prepare a state on the time registers
\begin{align}
&\frac 1{\sqrt{\lambda_x^2 + (e-1)^2\timesteps\lambda_b^2\Delta t^2}}  \left[ \lambda_x \ket{0} + (e-1)\lambda_b\Delta t \sum_{m=1}^\timesteps \ket{m} \right].
\end{align}
Then a preparation of $\ket{\bm{x}(t_0)}$ controlled on $\ket 0$ and a preparation of $\ket{\bm{v}_m}$ controlled on $\ket{m}$ gives a state of the form
\begin{equation}
    \frac 1{\sqrt{\lambda_x^2 + (e-1)^2\timesteps\lambda_b^2\Delta t^2}} \left[ \ket{0}\ket{\bm{x}(t_0)} + \sum_{m=1}^\timesteps \ket{m} \ket{\bm{v}_m} \right].
\end{equation}
The amplitude for success is then
\begin{equation}\label{eq:amsucc}
    \sqrt{\frac{\lambda_x^2 + \sum_{m=1}^\timesteps \|\bm{v}_m\|^2}{\lambda_x^2 + (e-1)^2\timesteps\lambda_b^2\Delta t^2}} \ge \frac{\min_m \|\bm{v}_m\|}{(e-1)\lambda_b\Delta t} .
\end{equation}

The expression here is in terms of $\bm{v}_m$, which will be constructed with the truncated series and discretised integrals, but the parameters will be chosen so that the error in each is no more than $\epsilon x_{\max}/\timesteps$.
As a result
\begin{equation}
    \|\bm{v}_m\| \ge \|\bm{v}(m\Delta t,(m-1)\Delta t)\| - \epsilon x_{\max}/\timesteps \, .
\end{equation}
It would be expected that the correction here would typically be small enough to ignore.
If the norm of $\bm{b}$ does not vary too much and there are not cancellations in the integral, then the expression is proportional to $b_{\max}\Delta t$.
Since it would be expected that $\lambda_b$ is comparable to $b_{\max}$, in such a case the amplitude for success would be at least a constant.
For full generality we use the full expression rather than considering typical cases.

The preparation of $\mathcal{B}$ needs to be performed twice in the walk step used in the solution of the linear equations using the approach in Ref.~\cite{QLSP}.
That would lead to the square of this amplitude for success in the overall amplitude for this walk step.
To avoid this square factor, we can perform amplitude amplification on the state preparation to boost its amplitude to $\Theta(1)$.
The number of steps of amplitude amplification will be the inverse of this amplitude, or
\begin{equation}\label{eq:stateaa}
    \bigO{\frac{\lambda_b \Delta t}{\min_m\|\bm{v}(m\Delta t,(m-1)\Delta t)\| - \epsilon x_{\max}/\timesteps}} .
\end{equation}
Usually the amplitude will be unknown, but it can be determined at the beginning of the procedure with a logarithmic overhead.
That amplitude estimation is considerably less costly than the cost of solving the linear equations, which has an extra factor of the condition number $\kappa_{\mathcal{A}}$.

There is also the initial preparation needed for the time register for the state.
This can be performed as follows.
First perform a rotation on an ancilla qubit to obtain the appropriate weighting between $m=0$ and $m=1,\ldots,\timesteps$.
If $\timesteps$ is a power of 2, then controlled on this qubit being $\ket{0}$ we perform Hadamards on the $\log\timesteps$ qubits for the time, giving values from $0$ to $\timesteps-1$.
In the more general case where $\timesteps$ is not a power of 2, one can perform a controlled preparation over a number of basis states that is not a power of 2, which has complexity $\bigO{\log\timesteps}$ \cite{Sanders2020}.

We can perform CNOTs from the ancilla qubit to the remaining qubits for the time, so that if this control qubit was $\ket{1}$ then we have all ones.
Then adding 1 in a modular way with the ancilla as the most significant bit, if the ancilla was $\ket{1}$ then we get all zeros, and if it was $\ket{0}$ we get values from $1$ to $\timesteps$.
The complexity of this procedure is $\bigO{\log(1/\epsilon)}$ for the initial rotation, and $\bigO{\log\timesteps}$ for the remaining steps.
This complexity can be disregarded because it is smaller than many other parts of the procedure.

\subsection{Time-independent block encoding}
The time-independent case is substantially simplified over the time-dependent case.  First, the block encodings are simplified to
\begin{align}
    _A\!\bra{0} U_A \ket{0}_A &= \frac 1{\lambda_A} ,\\
    U_b \ket{0} &= \frac 1{\lambda_b} \ket{\bm{b}}, \\
     U_x \ket{0} &= \frac 1{\lambda_x} \ket{\bm{x}_0},
\end{align}
where there is now no time register for the input and the preparation of $\ket{\bm{b}}$ is just via a unitary operation.
We now omit the $s$ subscript because there is just the one register.
This means that now $\lambda_b=\| \bm{b} \|$.
Again the block encoding can use Eq.~\eq{Asum}, except with $V_m$ replaced with the Taylor series $V$ as given in Eq.~\eq{Tayser}, and $\bm{v}_m$ replaced with $\bm{v}$ given in Eq.~\eq{bser}.

The block encoding proceeds in exactly the same way as before, except that there is no need for the time integrals.
This means that we need to prepare a register with $k$ for both the Taylor series and the series for $\bm{v}$.
Again, if $\lambda_A \Delta t\le 1$, then the Taylor series is block encoded with a factor of at least $1/e$.

We can also circumvent the problem we have in the time-dependent case that there are possible cancellations in $\bm{v}_m$.
That is, we have
\begin{align}
    \| \bm{v} \| &\ge \| \bm{b} \| \Delta t \left\{ 1 - \sum_{k=2}^K \frac{(\|A\|\Delta t)^{k-1}}{k!}\right\}\nn
    &> \| \bm{b} \| \Delta t \left\{ 1 - \sum_{k=2}^\infty \frac{(\|A\|\Delta t)^{k-1}}{k!} \right\} \nn
    &= \| \bm{b} \| \Delta t \left\{ 1 - \frac{\left[ \exp(\|A\|\Delta t) - (1+\|A\|\Delta t)\right]}{\|A\|\Delta t} \right\} \nn
    &\ge (3-e) \| \bm{b} \| \Delta t \, ,
\end{align}
where we have used $\|A\|\Delta t \le \lambda_A \Delta t \le 1$.
We would then find that the amplitude for success of the state preparation would, using Eq.~\eq{amsucc}, be given by
\begin{equation}
    \sqrt{\frac{\lambda_x^2 + \sum_{m=1}^\timesteps \|\bm{v}\|^2}{\lambda_x^2 + (e-1)^2\timesteps\lambda_b^2\Delta t^2}} \ge \frac{(3-e) \| \bm{b} \| \Delta t}{(e-1)\lambda_b\Delta t} = \frac{3-e}{e-1} \, .
\end{equation}
Thus the amplitude for success of the state preparation is at least a constant here, so cannot affect the scaling of the complexity and can be omitted in the $\mathcal{O}$ expressions.

\section{Complexity of solution}
\label{sec:final}

\subsection{Time-dependent complexity}

Before giving the overall complexity, we give further discussion of the choice of $K$.
As explained above, the error due to the truncation is given in Eq.~\eqref{eq:errorbnd}, and the error in the first term (due to the homogeneous component of the solution) can be appropriately bounded by choosing $K$ as in Eq.~\eqref{eq:choiceK}.

The complete expression for the error in Eq.~\eqref{eq:errorbnd} accounts for the inhomogeneity in the second term.
The expression for the inhomogeneous error can also be bounded with $\lambda_A$ replacing $A_{\max}$ because $\lambda_A\ge A_{\max}$, which gives this term as
\begin{equation}
    \bigO{ \frac{\lambda_A^K \Delta t^{K+1} b_{\max}}{(K+1)!} } .
\end{equation}
For our algorithm we will be choosing $\Delta t$ such that $\lambda_A \Delta t\le 1$, which means that this error is bounded as
\begin{equation}
    \bigO{ \frac{b_{\max}}{(K+1)!\lambda_A} } .
\end{equation}
Then we can choose $K$ to make this error no larger than $\epsilon x_{\max}/r$ by taking
\begin{equation}
    K = \bigO{\frac{\log(r b_{\max}/(\lambda_A x_{\max}\epsilon))}{\log\log(r b_{\max}/(\lambda_A x_{\max}\epsilon))}} .
\end{equation}
With $\Delta t\approx 1/\lambda_A$, we have $r=\bigO{\lambda_A T}$, which gives
\begin{equation}
    K = \bigO{\frac{\log(T b_{\max}/(x_{\max}\epsilon))}{\log\log(T b_{\max}/(x_{\max}\epsilon))}} .
\end{equation}
Using $r=\bigO{\lambda_A T}$ in Eq.~\eqref{eq:choiceK} for the choice of $K$ for the homogeneous error gives
\begin{equation}
    K = \bigO{\frac{\log(\lambda_A T/\epsilon)}{\log\log(\lambda_A T/\epsilon)}} .
\end{equation}
Since we need to ensure that both error terms in Eq.~\eqref{eq:errorbnd} are appropriately bounded, we should take the maximum of these two scalings of $K$, which can be expressed as
\begin{equation}\label{eq:chooseK2}
    K = \bigO{\frac{\log(\lambda_{Ax} T/\epsilon)}{\log\log(\lambda_{Ax} T/\epsilon)}} ,
\end{equation}
where
\begin{equation}
    \lambda_{Ax} = \max\left( \lambda_A,{b_{\max}}/{x_{\max}} \right) \, .
\end{equation}

We will be measuring the error between pure states via the 2-norm distance, but also need to account for cases where the approximate state generated is not pure.
We therefore use the Bures-Wasserstein metric \cite{BHATIA2019165}
\begin{equation}
    d_{BW}(\rho,\sigma) = \sqrt{{\rm Tr}(\rho) + {\rm Tr}(\sigma) - 2 {\rm Tr}\sqrt{\sqrt{\rho}\sigma\sqrt{\rho}}} \, .
\end{equation}
This measure is for states that are not normalised, which is why it includes the traces.
In the case that they are normalised this measure is the Bures distance.
For the case where the two states are pure, so $\rho=\ket{\psi}\bra{\psi}$ and $\sigma=\ket{\phi}\bra{\phi}$, then provided $\braket{\psi}{\phi}$ is real (and positive) we obtain
\begin{equation}
    d_{BW}(\rho,\sigma) = \| \ket{\psi} - \ket{\phi} \| \, .
\end{equation}
Therefore this measure can be regarded as a generalisation to mixed states of the 2-norm distance.
Another feature of this measure is that $d_{BW}^2$ is convex in $\rho$.
This property follows from concavity of the fidelity and square root.

The overall complexity of the quantum algorithm for the time-dependent differential equations can be described as in the following theorem.

\begin{theorem}
We are given an ordinary linear differential equation of the form
\begin{equation}
\bm{\dot x}(t) = A(t)\bm{x}(t)+\bm{b}(t), \qquad \bm{x}(0) = \bm{x}_0,
\end{equation}
where $\bm{b}(t) \in \mathbb{C}^{N}$ is a vector function of $t$, $A(t)\in \mathbb{C}^{N\times N}$ is a coefficient matrix with non-positive logarithmic norm, and 
$\bm{x}(t)\in \mathbb{C}^{N}$ is the solution vector as a function of $t$.
The parameters of the differential equation are provided via unitaries $U_A$, $U_b$, and $U_x$ with known $\lambda_A,\lambda_b,\lambda_x$ such that
\begin{align}
    _A\!\bra{0} U_A \ket{0}_A\ket{n_t} &= \frac 1{\lambda_A} A(t) \ket{n_t},\\
    _b\!\bra{0} U_b \ket{0}_b\ket{n_t}\ket{0}_s &= \frac 1{\lambda_b} \ket{n_t} \ket{\bm{b}(t)}_s, \\
     U_x \ket{0}_s &= \frac 1{\lambda_x} \ket{\bm{x}_0}_s.
\end{align}
A quantum algorithm can provide an approximation $\hat{\bm{x}}$ of the solution $\ket{\bm{x}(T)}$ satisfying $d_{BW}\left(\hat{\bm{x}},\ket{\bm{x}(T)}\bra{\bm{x}(T)}\right) \le \epsilon x_{\max}$ using an average number
\begin{align}
    &\bigO{\solrat \lambda_{A} T \log\left( 1/\epsilon \right)} && {\rm calls~to~}U_b,U_x,\\
    &\bigO{\solrat \lambda_{A} T \log\left( 1/\epsilon \right)\log\left(\frac{\lambda_{Ax} T}\epsilon\right)} && {\rm calls~to~}U_A,\\
    &\mathcal{O}\left(\solrat \lambda_{A} T \log\left( 1/\epsilon \right)\log\left(\frac{\lambda_{Ax} T}\epsilon\right)\left[ \log \left( \frac{T \dervbnd}{\lambda_{A}\epsilon}\right) + \log\left(\frac{\lambda_{A} T}\epsilon\right) \right] \right) \quad && {\rm additional~gates},
\end{align}
where $\lambda_{Ax} = \max\left(\lambda_A,{b_{\max}}/{x_{\max}}\right)$, given constants satisfying
\begin{align}
    \solrat &\ge \frac{x_{\max}}{\left\| \bm{x}(T) \right\|}  \frac{\lambda_b/\lambda_{A}}{\min_{m} \|\bm{v}(m\Delta t,(m-1)\Delta t)\|-\epsilon x_{\max}/(\lambda_{A} T)},\\
    \dervbnd &\ge \left( 1+ \frac{Tb_{\max}}{\lambda_A x_{\max}} \right)\max_{t\in[0,T]} \| A'(t) \|+ \frac{\max_{t\in[0,T]} \| \bm{b}'(t) \|}{x_{\max}},\\
    x_{\max} &\ge \max_{t\in[0,T]}\left\| \bm{x}(t) \right\|, \\
    b_{\max} &\ge \max_{t\in[0,T]}\left\| \bm{b}(t) \right\|.
\end{align}
Here
\begin{align}
\bm{v}(t,t_0) &\coloneqq \sum_{k=1}^{\infty} \int_{t_0}^t dt_1 \int_{t_0}^{t_1} dt_2 \cdots \int_{t_0}^{t_{k-1}} dt_{k}\, A(t_1) A(t_2) \cdots A(t_{k-1})\bm{b}(t_{k}) \, .
\end{align}
and
\begin{equation}
   \Delta t = \frac{T} {\left\lceil \lambda_A T\right\rceil} .
\end{equation}
\end{theorem}

The condition that the logarithmic norm is non-positive ensures that the differential equation is dissipative.
We use the convention that state vectors such as $\ket{\bm{x}_0}_s$ are the unnormalised form with amplitudes exactly equal to the coefficients of the corresponding complex vector.
For the unnormalised solution state $\ket{\hat{\bm{x}}}$, this is defined in the sense that for the normalised quantum state, there exists some constant that we can multiply it by to give $\ket{\hat{\bm{x}}}$ that satisfies the accuracy constraint.

We require that we are given upper bounds $\solrat$, $\dervbnd$, and so on, because it may be too difficult to determine the maxima and minima required exactly.
For completeness we have given a complicated expression for $\solrat$ to account for the steps of amplitude amplification needed, but in practice if the solution does not decay significantly, and $\bm{b}(t)$ does not vary such that it cancels and gives small $\bm{v}$, then it can be expected that $\solrat=\bigO{1}$ and can be ignored in the scaling of the complexity.

For the count of additional gates, we are not allowing arbitrary precision rotations.
Instead it would be for a fixed set of gates, such as Toffolis (or $T$ gates) plus Clifford gates.
The complexity would be equivalent (up to a constant factor) to the non-Clifford count that is often used in quantifying complexities for algorithms intended for error-corrected quantum computers with magic-state factories.

\begin{proof}
The principle we use for the solution is to express the solution in terms of a Dyson series over many time steps encoded in a system of linear equations as shown in Eq.~\eqref{eq:lineq}.
We then use the linear equation solver of Ref.~\cite{QLSP} which has complexity $\mathcal{O}\left(\kappa_{\mathcal{A}}\log(1/\epsilon)\right)$, in terms of calls to the block-encoded matrix $\mathcal{A}$ and vector $\mathcal{B}$.
We then need to perform amplitude amplification to obtain the solution at the final time from the vector with the solution over all times.
In order to determine the complexity, we therefore need to determine $\kappa_{\mathcal{A}}$, the complexity of block encoding $\mathcal{A},\mathcal{B}$, and the amplitude of the component of the solution at the final time.
We will show that the factor of $\solrat$ comes from the complexity of that amplitude amplification, as well as amplitude amplification as part of preparing $\mathcal{B}$.
Moreover, we will show that $\lambda_A T$ comes from the factor of $\kappa_{\mathcal{A}}$,
and $\log\left({\lambda_{Ax} T}/\epsilon\right)$ comes from the order $K$ used in the Dyson series for block encoding $\mathcal{A},\mathcal{B}$.

The simplest part is the condition number.
The condition number is $\kappa_{\mathcal{A}} = \bigO{\totalrows}$ as in Eq.~\eqref{eq:condno}.
Here, $\totalrows$ is the total number of time steps, including those at the end where the solution is held fixed at the final time.
This is the standard approach to ensure that there is a reasonable amplitude for the solution at the final time.
We use an equal number of time steps for the time evolution $r$ as those where the solution is held fixed, so that $\totalrows=2r$ and $\kappa_{\mathcal{A}} = \bigO{r}$.
Because $r=T/\Delta t$, to minimise the complexity we should choose $\Delta t$ as large as possible.
The restriction limiting how large $\Delta t$ can be is that $\lambda_A \Delta t\le 1$, which we use to ensure that the $\lambda$-value for the block encoding of $\mathcal{A}$ is $\bigO{1}$.
That means we should take $\Delta t\approx 1/\lambda_A$ so $r=\lambda_A T$, and therefore $\kappa_{\mathcal{A}} = \bigO{\lambda_A T}$ as claimed.

The complexity of solving the system of linear equations is $\bigO{\lambda_A T\log(1/\epsilon_{\rm QLSP})}$ in terms of calls to block encodings of $\mathcal{A},\mathcal{B}$, where $\epsilon_{\rm QLSP}$ is the allowable error in the solution of the system of linear equations.
For the complexity we need to relate that error to the allowable error in the solution.

To achieve that, let us denote the approximate solution obtained at time $t$ by $\ket{\tilde{\bm{x}}(t)}$, to be compared with the correct solution of the linear equations $\ket{\bm{x}(t)}$.
Because the solution of the linear equations is approximate, the components of the approximate solution vector $\widetilde{\mathcal{X}}$ need not be equal in the second half as they are for the exact solution.
We are using $\ket{\tilde{\bm{x}}(t)}$ for the components of $\widetilde{\mathcal{X}}$, so to account for the unequal components in the notation, we use $\ket{\tilde{\bm{x}}(t)}$ with values of $t>T$, and similarly for $\bm{x}$.

When the solution of the linear equations is accurate to within error $\epsilon_{\rm QLSP}$, it means that
\begin{align}
\sqrt{\sum_{j=r+1}^\totalrows \left\| \ket{\tilde{\bm{x}}(j \Delta t)} - \ket{\bm{x}(T)} \right\|^2}
& \le
\sqrt{\sum_{j=0}^\totalrows \left\| \ket{\tilde{\bm{x}}(j \Delta t)} - \ket{\bm{x}(j \Delta t)} \right\|^2} \nn
&\le \epsilon_{\rm QLSP} \sqrt{\sum_{j=0}^\totalrows \left\| \ket{\bm{x}(j \Delta t)} \right\|^2}\nn 
&\le \epsilon_{\rm QLSP} \, x_{\max} \sqrt{\totalrows+1} .
\end{align}
Therefore, the approximate solution $\hat{\bm{x}}$ will satisfy
\begin{align}
d_{BW}\left(\hat{\bm{x}},\ket{\bm{x}(T)}\bra{\bm{x}(T)}\right) 
&\le \sqrt{\frac 1{r} \sum_{j=r+1}^\totalrows d_{BW}^2\left( \ket{\tilde{\bm{x}}(j \Delta t)}\bra{\tilde{\bm{x}}(j \Delta t)} , \ket{\bm{x}(T)}\bra{\bm{x}(T)} \right)} \nn
&= \frac 1{\sqrt{r}} \sqrt{\sum_{j=r+1}^\totalrows \left\| \ket{\tilde{\bm{x}}(j \Delta t)} - \ket{\bm{x}(T)} \right\|^2} \nn
&\le \frac 1{\sqrt{r}} \epsilon_{\rm QLSP} \, x_{\max} \sqrt{\totalrows+1} .
\end{align}
Because we take $\totalrows=2r$ we can take $\epsilon_{\rm QLSP}\propto \epsilon$ and the desired precision will be obtained.
Note that there will be error in other steps in the algorithm so the error in solving the linear equations will be taken to be some fraction of the total allowable error $\epsilon$.
As usual in analysis of this type, the constant factor can be omitted when giving complexity scalings.
Therefore the complexity of solving the linear equations can be given as $\bigO{\lambda_A T\log(1/\epsilon)}$ in terms of calls to block encodings of $\mathcal{A},\mathcal{B}$.

The complexity of block encoding $\mathcal{A}$ can be determined from the complexity of block encoding the Dyson series in Eq.~\eqref{eq:truncDys}, with other costs being trivial.
That Dyson series uses calls to the oracles $U_A$, but not $U_b,U_x$.
Because the Dyson series is truncated at order $K$, it can be block encoded with $K$ calls to $U_A$, with $K$ as given in Eq.~\eqref{eq:chooseK2}.
That gives the factor of $\log\left({\lambda_{Ax} T}/\epsilon\right)$ in the complexity for calls to $U_A$.

The complexity of block encoding $\mathcal{B}$ is similar, except we also need to account for calls to $U_b,U_x$, and we also need to account for the complexity of amplitude amplification.
The form of the Dyson series used to prepare $\mathcal{B}$ is given in Eq.~\eqref{eq:bDys}.
The block encoding of this Dyson series uses $K-1$ calls to $U_A$, and a single call to $U_b$.
Moreover, a single call to $U_x$ is used to prepare the component of $\mathcal{B}$ with the initial state.
Therefore we have the factor of $K$ for calls to $U_A$, but \emph{not} to $U_b,U_x$.

Next, the number of steps of amplitude amplification for preparing the state for $\mathcal{B}$ is given in Eq.~\eqref{eq:stateaa}.
Using $\Delta t \propto 1/\lambda_A$ gives the number of steps as
\begin{equation}\label{eq:secondfac}
    \bigO{\frac{\lambda_b /\lambda_A}{\min_m\|\bm{v}(m\Delta t,(m-1)\Delta t)\| - \epsilon x_{\max}/(\lambda_A T)}} .
\end{equation}
That is the second factor in the expression for $\solrat$.

Next we consider the number of steps of amplitude amplification needed to obtain the solution at the final time.
This amplitude amplification is performed on the state obtained from the linear equation solver.
The amplitude of the state at the final time is given by
\begin{align}
    \sqrt{\frac{\sum_{j=r+1}^\totalrows \| \ket{\tilde{\bm{x}}(j \Delta t)} \|^2}{\sum_{j=0}^\totalrows \| \ket{\tilde{\bm{x}}(j \Delta t)} \|^2}} ,
\end{align}
which implies an average number of steps of amplitude amplification
\begin{align}\label{eq:finaasteps}
    \bigO{\sqrt{\frac{\sum_{j=0}^\totalrows \| \ket{\tilde{\bm{x}}(j \Delta t)} \|^2}{\sum_{j=r+1}^\totalrows \| \ket{\tilde{\bm{x}}(j \Delta t)} \|^2}}} .
\end{align}
For the numerator we obtain
\begin{align}
\sum_{j=0}^\totalrows \| \ket{\tilde{\bm{x}}(j \Delta t)} \|^2 
&\le \sum_{j=0}^\totalrows \left(\| \ket{{\bm{x}}(j \Delta t)} \|
+ \| \ket{\tilde{\bm{x}}(j \Delta t)}-\ket{{\bm{x}}(j \Delta t)} \|\right)^2 \nn
&\le (\totalrows+1) x_{\max}^2 + 2x_{\max}\sum_{j=0}^\totalrows \| \ket{\tilde{\bm{x}}(j \Delta t)}-\ket{{\bm{x}}(j \Delta t)} \|\nn
&\quad+ \sum_{j=0}^\totalrows \| \ket{\tilde{\bm{x}}(j \Delta t)}-\ket{{\bm{x}}(j \Delta t)} \|^2 \nn
&\le (\totalrows+1) x_{\max}^2 + 2x_{\max}\sqrt{\totalrows+1}\nn 
&\quad \times \sqrt{\sum_{j=0}^\totalrows \| \ket{\tilde{\bm{x}}(j \Delta t)}-\ket{{\bm{x}}(j \Delta t)} \|^2} + \epsilon_{\rm QLSP}^2 \, x_{\max}^2 (\totalrows+1) \nn
&\le (\totalrows+1) x_{\max}^2 + 2x_{\max}\epsilon_{\rm QLSP} \, x_{\max} (\totalrows+1)+ \epsilon_{\rm QLSP}^2 \, x_{\max}^2 (\totalrows+1) \nn
&= (\totalrows+1) x_{\max}^2 ( 1 + \epsilon_{\rm QLSP} )^2 \nn
&= \bigO{r x_{\max}^2}.
\end{align}
For the denominator of Eq.~\eqref{eq:finaasteps} we similarly obtain
\begin{align}
{\sum_{j=r+1}^\totalrows \| \ket{\tilde{\bm{x}}(j \Delta t)} \|^2} + {\sum_{j=r+1}^\totalrows \| \ket{\tilde{\bm{x}}(j \Delta t)}-\ket{{\bm{x}}(j \Delta t)} \|^2} 
&\ge {\sum_{j=r+1}^\totalrows \| \ket{{\bm{x}}(j \Delta t)} \|^2} \nn
\sqrt{\sum_{j=r+1}^\totalrows \| \ket{\tilde{\bm{x}}(j \Delta t)} \|^2} + \sqrt{\sum_{j=r+1}^\totalrows \| \ket{\tilde{\bm{x}}(j \Delta t)}-\ket{{\bm{x}}(j \Delta t)} \|^2} 
&\ge \sqrt{\sum_{j=r+1}^\totalrows \| \ket{{\bm{x}}(j \Delta t)} \|^2} \nn
\sqrt{\sum_{j=r+1}^\totalrows \| \ket{\tilde{\bm{x}}(j \Delta t)} \|^2} + \sqrt{\totalrows+1}\epsilon_{\rm QLSP} x_{\max} &\ge \sqrt r \|{\bm{x}}(T)\| \nn
\sqrt{\sum_{j=r+1}^\totalrows \| \ket{\tilde{\bm{x}}(j \Delta t)} \|^2}
    &\ge \sqrt{r} \left( \|{\bm{x}}(T)\| - \epsilon x_{\max}/2 \right) ,
\end{align}
where in the last line we have used $\epsilon_{\rm QLSP}$ as suitable fraction of $\epsilon$.
Now, provided $\|{\bm{x}}(T)\| \ge \epsilon x_{\max}$ we obtain
\begin{equation}
    \sum_{j=r+1}^\totalrows \| \ket{\tilde{\bm{x}}(j \Delta t)} \|^2 \ge \frac r4 \|{\bm{x}}(T)\|^2 \, .
\end{equation}
If $\|{\bm{x}}(T)\| < \epsilon x_{\max}$ then this is a pathological case where the size of the solution is smaller than the precision required in the algorithm.
For any output state we can give $\hat{\bm{x}}=0$ and the solution will have the desired precision.
Therefore, excluding that pathological case we obtain
\begin{align}
\sqrt{\frac{\sum_{j=0}^\totalrows \| \ket{\tilde{\bm{x}}(j \Delta t)} \|^2}{\sum_{j=r+1}^\totalrows \| \ket{\tilde{\bm{x}}(j \Delta t)} \|^2}} &=\bigO{\frac{\sqrt{r x_{\max}^2}}{\sqrt{r\|{\bm{x}}(T)\|^2}}}\nn 
&= \bigO{\frac{x_{\max}}{\|{\bm{x}}(T)\|}}.
\end{align}
The overall factor in the complexity required is this number of steps of amplitude amplification for obtaining the final time, multiplied by the number of steps of amplitude amplification for preparing $\mathcal{B}$ given in Eq.~\eqref{eq:secondfac}.
That product is the lower bound for $\solrat$ given in the statement of the theorem.

A minor subtlety in the use of amplitude amplification on the result of the quantum linear equations algorithm is that the algorithm as given in Ref.~\cite{QLSP} uses measurements to ensure the correct result of filtering, which would not be compatible with amplitude amplification.
However, it is trivial to just use omit that measurement, and perform the amplitude amplification jointly on the successful result of the filtering and obtaining the final time in the solution state.
That introduces no further factors in the complexity.

Hence, we have proven the factor $\solrat$ in the complexities for $U_A$ and $U_b,U_x$, with the factor $\lambda_A T$ arising from the factor of $\kappa_{\mathcal{A}}$ in the complexity for solving linear equations, $\log(1/\epsilon)$ coming from the factor in the complexity of solving linear equations, and the factor of $\log\left({\lambda_{\color{blue}Ax} T}/\epsilon\right)$ for the $U_A$ complexity coming from the order $K$ used in the Dyson series.
It just remains to account for the number of additional gates.

The majority of the complexity for the additional gates is in constructing the registers for the time.
We need to construct $K$ registers for the time, and sort in order to obtain the correctly ordered time registers.
 Similarly to the case for Hamiltonian simulation in Ref.~\cite{Kieferova2019}, the complexity of these steps is then $\bigO{K\log K \log\smalltime}$, where the factor of $K\log K$ is for the number of steps in the sort, and the factor of $\log\smalltime$ is for the number of qubits to store each time.
Taking $K$ as in Eq.~\eqref{eq:chooseK2}, we obtain $K\log K=\bigO{\log(\lambda_{Ax} T/\epsilon)}$.

That accounts for the factor of $\log(\lambda_{Ax} T/\epsilon)$ before the square brackets given for the number of additional gates.
Next we consider the contribution to the complexity from the factor of $\log \smalltime$.
The value of $\smalltime$ needs to be chosen to be large enough such that the error from discretisation of the integrals in the Dyson series is no larger than $\epsilon/r$ for each of the $r$ time segments.
Exactly the same derivation as for the time-dependent Hamiltonian in Ref.~\cite{Kieferova2019} holds, where the result is given in Eq.~(58) and the following text of that work.
For completeness we give the derivation of the error in Appendix \ref{app:timedisc}, which gives
\begin{equation}
\label{eq:Dysdicer}
\norm{W_K\left(t_\beta,t_\alpha \right) - \widetilde{W}_K\left(t_\beta,t_\alpha \right)} \in \bigO{\frac{(T/\timesteps)^2}{\smalltime}\max_{t\in [0,T]}\| A'(t) \|}
\end{equation}
where $W_K(t_\alpha,t_\beta)$ is given by Eq.~\eq{trunc}, $t_\alpha-t_\beta=T/r$ and
\begin{equation}
\widetilde{W}_{K} = \sum_{k=0}^{K}\frac{\left(T/r\right)^k}{M^k k!}\sum_{j_1,\cdots, j_k=0}^{M-1}\mathcal{T} A(t_{j_k}) \cdots A(t_{j_1}) . 
\end{equation}
This expression corresponds to discretising the integral over each time variable to approximate it by a sum with $M$ terms.
In Eq.~\eq{Dysdicer} we also have $A'(t)$ which is the time derivative of $A(t)$.

Similarly, it is easy to see that 
\begin{align}
\label{eq:error_part}
\norm{\bm{v}_K\left(t_\beta,t_\alpha\right) - \tilde{\bm{v}}_K \left(t_\beta,t_\alpha\right)}  
\in \mathcal{O}\left(\frac{(T/\timesteps)^2}{\smalltime}\max_{t\in [0,T]}\| \bm{b}'(t) \|+ \frac{(T/\timesteps)^3 b_{\max}}{\smalltime}\max_{t\in [0,T]}\| A'(t) \| \right),
\end{align}
for the error in the integral for the driving term due to time discretisation of Eq.~\eq{part_sol}. The first term in Eq.~\eq{error_part} comes from $k=1$, and the second comes from $k=2$.
See Appendix \ref{app:timedisc} for details of the derivation.

Now, because the error in each evolution operator $W_K$ should be no larger than $\epsilon/r$, the expression for the error given in Eq.~\eqref{eq:Dysdicer} implies that we should choose
\begin{equation}\label{eq:cond1}
    \smalltime = \Omega\left( \frac{T^2}{r\epsilon}\max_{t\in [0,T]}\| A'(t) \| \right) = \Omega\left( \frac{T}{\lambda_{A}\epsilon}\max_{t\in [0,T]}\| A'(t) \| \right) .
\end{equation}
Recall that the error in $W_K$ needs to be no larger than $\epsilon/r$, because we multiply this operator with the vector $\bm{x}$ to give a final error in the vector no larger than $\epsilon x_{\max}$.
The error in $\bm{v}_K$ should be no larger than $\epsilon x_{\max}/r$, because it directly translates to error in $\bm{x}$.
The requirement to bound the first term in Eq.~\eqref{eq:error_part} implies that we should take
\begin{equation}\label{eq:cond2}
    \smalltime = \Omega\left( \frac{T}{\lambda_A\epsilon \, x_{\max}}\max_{t\in [0,T]}\| \bm{b}'(t) \| \right) .
\end{equation}
In order to bound the second term in Eq.~\eqref{eq:error_part}, we should take
\begin{equation}\label{eq:cond3}
    \smalltime = \Omega\left( \frac{T^2 b_{\max}}{\lambda_A^2\epsilon \, x_{\max}}\max_{t\in [0,T]}\| A'(t) \| \right) .
\end{equation}
The smallest $\smalltime$ satisfying all three conditions \eqref{eq:cond1}, \eqref{eq:cond2}, and \eqref{eq:cond3}, is then
\begin{equation}
    \smalltime = \bigO{\frac{T\dervbnd}{\lambda\epsilon}} ,
\end{equation}
with our bound on $\dervbnd$.
Taking the log of this expression gives the first term in square brackets for the number of additional gates in the theorem.

Lastly, we consider the complexity of the rotations used for preparing the $k$ register.
As explained above, these give complexity $\bigO{K\log K\log(\timesteps/\epsilon)}$.
Using Eq.~\eqref{eq:chooseK2} for $K$ gives $K\log K=\bigO{\log(\lambda_{Ax} T/\epsilon)}$, then
$\bigO{\log(\timesteps/\epsilon)}$ gives another factor of $\bigO{\log(\lambda_{A} T/\epsilon)}$, which is given as the second term in the square brackets for the gate complexity.
In practice it would be expected that this term is smaller.

There are also operations needed on the approximately $\log \timesteps$ qubits for preparing the time registers, but this complexity is no more than $\bigO{K}$ for each time step.
That is smaller than the other complexities so can be omitted in the order scaling.
\end{proof}

A similar form of scaling of complexity was given for the spectral method in \cite{Childs2020}.
The statement of the complexity in that work gave the log factors in Theorem 1 just as polylog, but the true scaling can be seen in their proof in Eq.~(8.15).
There the complexity is given as proportional to $n^4$ times a polylog factor coming from solving systems of linear equations.
Their quantity $n$ scales as the log of $1/\epsilon$, $T$, and the derivatives of the solution.

Therefore the approach of \cite{Childs2020} gives scaling as the fourth power of these logs, times whatever factor is obtained from the solution of linear systems.
In that work the linear equation solver of Childs, Kothari, and Somma \cite{CKS} was used, which gives a complexity superlinear in $\log(1/\epsilon)$.
The approach of \cite{Childs2020} can also be used with the optimal solver of \cite{QLSP} with just linear complexity in $\log(1/\epsilon)$, as we do here.

Assuming that optimal solver for a fair comparison, the scaling of the number of calls to a block encoding of $A(t)$ in terms of $T$ and $\epsilon$ in \cite{Childs2020} is as $\log^4(T/\epsilon)\log(1/\epsilon)$ (with the second log from the solver).
In contrast our complexity is $\log(T/\epsilon)\log(1/\epsilon)$.
Moreover, \cite{Childs2020} has the same $\log^4(T/\epsilon)\log(1/\epsilon)$ complexity in calls to block encodings of $\bm{b}(t)$, whereas our complexity is greatly improved to just the $\log(1/\epsilon)$ factor from the linear equation solver.
That is because $\bm{b}(t)$ is only needed once in the Dyson series solution.

A further improvement in our approach is that the number of oracle calls is independent of derivatives of $A(t)$ and $\bm{b}(t)$.
In \cite{Childs2020} the complexity scales as the fourth power of the log of arbitrary-order derivatives through their variable $g'$.
In our approach there is dependence on derivatives only in the number of additional gates (needed for addressing time registers), and that complexity only depends on the first-order derivatives.

\subsection{Time-independent complexity}

In the time-independent case the complexity becomes as follows.

\begin{theorem}
We are given an ordinary linear differential equation of the form
\begin{equation}
\bm{\dot x}(t) = A\bm{x}(t)+\bm{b}, \qquad \bm{x}(0) = \bm{x}_0,
\end{equation}
where $\bm{b} \in \mathbb{C}^{N}$ is a vector function of $t$, $A\in \mathbb{C}^{N\times N}$ is a coefficient matrix with non-positive logarithmic norm, and 
$\bm{x}(t)\in \mathbb{C}^{N}$ is the solution vector as a function of $t$.
The parameters of the differential equation are provided via unitaries $U_A$, $U_b$, and $U_x$ such that
\begin{align}
    _A\!\bra{0} U_A \ket{0}_A &= \frac 1{\lambda_A} A,\\
    U_b \ket{0} &= \frac 1{\lambda_b} \ket{\bm{b}}, \\
    U_x \ket{0} &= \frac 1{\lambda_x} \ket{\bm{x}_0}.
\end{align}
A quantum algorithm can provide an approximation $\hat{\bm{x}}$ of the solution $\ket{\bm{x}(T)}$ satisfying $d_{BW}\left( \hat{\bm{x}} , \ket{\bm{x}(T)}\bra{\bm{x}(T)}\right) \le \epsilon x_{\max}$ using an average number
\begin{align}
    &\bigO{\solrat \lambda T \log\left( 1/\epsilon \right)} && {\rm calls~to~}U_b,U_x,\\
    &\bigO{\solrat \lambda T \log\left( 1/\epsilon \right)\log\left(\frac{\lambda_{Ax} T}\epsilon\right)} && {\rm calls~to~}U_A,\\
    &\bigO{\solrat \lambda T \log\left( 1/\epsilon \right)\log\left(\frac{\lambda_{Ax} T}\epsilon\right) \log\left(\frac{\lambda_{A} T}\epsilon\right)} \quad && {\rm additional~gates},
\end{align}
where $\lambda_{Ax} = \max\left(\lambda_A,{\| \bm{b}\|}/{x_{\max}}\right)$, given constants satisfying
\begin{align}
    \solrat &\ge \frac{x_{\max}}{\left\| \bm{x}(T) \right\|} ,\\
    x_{\max} &\ge \max_{t\in[0,T]}\left\| \bm{x}(t) \right\|.
\end{align}
\end{theorem}

\begin{proof}
The proof proceeds in exactly the same way as for the time-dependent case, except for a few minor amendments.
First, the state preparation succeeds with probability at least a constant, so the factor from this amplitude does not appear in $\solrat$.
The other difference in the proof is that the number of additional gates is greatly reduced.
We no longer need to perform any arithmetic on time registers, so the $\log \left( {T \dervbnd}/{\lambda\epsilon} \right)$ factor from the time-dependent cost is removed.
However, we do still need to perform rotations and controlled rotations to prepare the register with $k$.
As a result, we still have the extra factor of $\log\left({\lambda T}/\epsilon\right)$.
\end{proof}

\section{Conclusion}
\label{sec:conc}
We have given a quantum algorithm to solve a time-dependent differential equation in the sense that it outputs a quantum state with amplitudes encoding the solution vector.
This is a distinct method than that used for the time-independent case in \cite{BerryCMP17}, where different orders of the sum were encoded in successive lines of the block matrix.
In our approach the block matrix has each successive line encoding a new time step, and employs a block encoding of the Dyson series in a similar way as for Hamiltonian simulation in \cite{Kieferova2019}.

This approach makes the analysis of the complexity simpler than in \cite{BerryCMP17}, because it is not necessary to account for the extra lines in the encoding.
As in \cite{Krovi}, the expression for the complexity is simplified by using a condition on the logarithmic norm of $A(t)$.
This approach avoids needing $A(t)$ to be diagonalisable.
These techniques also give a significantly simplified result for the complexity in the time-independent case.

The complexity is, up to logarithmic factors, linear in the total evolution time $T$ and the constant $\lambda$ in the block encoding of $A(t)$.
By applying the optimal quantum linear equation solver of \cite{QLSP}, we obtain excellent scaling of the complexity in $\log(1/\epsilon)$.
The number of calls to encodings of the driving and initial state is linear in $\log(1/\epsilon)$, whereas the number of calls to the block encoding of $A(t)$ is quadratic in $\log(1/\epsilon)$.

The complexity only depends logarithmically on the first derivatives of $A(t)$ and the driving, and only the additional gates depend on these derivatives not the calls to the block encoding.
We have expressed the complexity in terms of maxima of these quantities.
It is also possible to express the complexity in terms of averages via the approach in \cite{Low2018}, though the derivation is considerably more complicated.
A bound could also be given in terms of the total variational distance as in \cite{Fang}.
That is useful when there are discontinuities.
Those bounds are just from a more careful analysis of the error in the discretisation of the integrals, rather than a different algorithm.

The way the result for the complexity is expressed is complicated by the need to account for the number of steps of amplitude amplification, which in turn depends on the norms of the $\bm{v}$ vectors.
This full expression is to account for pathological cases, which could potentially be ruled out by a more careful analysis to give a simplified expression.

Our result gives a significant improvement in the log factors over the spectral method from \cite{Childs2020}.
This raises the question of whether our complexity is optimal, or whether further improvements are possible.
We have log factors appearing from the optimal linear equation solver and from the order of the Dyson series, so any further improvement would require a radically different approach.

\section*{Acknowledgements}

We thank Maria Kieferova for helpful discussions.
DWB worked on this project under a sponsored research agreement with Google Quantum AI.
DWB is supported by Australian Research Council Discovery Projects DP190102633, DP210101367, and DP220101602.


\onecolumn\newpage
\appendix

\section{Error estimates for time discretisation}
\label{app:timedisc}

First we consider the error in the sum approximation of the integrals in the definition of $W_K$.
We can write $W_K$ as
\beq
W_K(t_\beta,t_\alpha) = \sum_{k=0}^{K} \frac 1{k!} \int_{t_\alpha}^{t_\beta} dt_1 \int_{t_\alpha}^{t_\beta} dt_2 \cdots \int_{t_\alpha}^{t_\beta} dt_k \, \mathcal{T} A(t_1) A(t_2) \cdots A(t_k) \, ,
\eeq
where $\mathcal{T}$ reorders the $A$ matrices in descending time order, and we are taking
$t_\alpha \in \{0,\Delta t,2\Delta t,\cdots, T-\Delta t\}$ and $t_\beta=t_\alpha+\Delta t$.
The sum approximation of $W_K$ is then
\begin{equation}
\widetilde{W}_K(t_\beta,t_\alpha) = \sum_{k=0}^{K}\frac{\delta t^k}{k!}\sum_{j_1,\cdots, j_k=0}^{\smalltime-1}\mathcal{T} A(t_{j_1}) \cdots A(t_{j_{k-1}})A(t_{j_k})\, ,
\end{equation}
where $\delta t=\Delta t/\smalltime$.
The sum can be rewritten as
\beq
\widetilde{W}_K(t_\beta,t_\alpha) = \sum_{k=0}^{K} \frac 1{k!} \int_{t_\alpha}^{t_\beta} dt_1 \int_{t_\alpha}^{t_\beta} dt_2 \cdots \int_{t_\alpha}^{t_\beta} dt_k \, \mathcal{T} A(t_{j_1}) A(t_{j_2}) \cdots A(t_{j_k}) \, ,
\eeq
where $t_{j_\ell}$ is $t_\ell$ rounded down to the nearest multiple of $\delta t$ as
\begin{equation}
    t_{j_\ell} = t_\alpha + \delta t\lfloor (t_\ell - t_\alpha)/\delta t\rfloor .
\end{equation}
This integral form of the sum is obvious, because each $t_{j_\ell}$ is constant over the interval $\delta t$.

Then we can write the error due to the time discretisation as
\begin{align}
    \norm{W_K(t_\beta,t_\alpha) - \widetilde{W}_K(t_\beta,t_\alpha)} &= \left\|\sum_{k=0}^K \frac 1{k!} \int_{t_\alpha}^{t_\beta}dt_1\int_{t_\alpha}^{t_\beta}dt_2 \cdots \int_{t_\alpha}^{t_\beta}dt_k \, \mathcal{T} \left[ A(t_1) A(t_2) \cdots A(t_k) \right.\right. \nn
    & \quad \left.\left. - A(t_{j_1}) A(t_{j_2}) \cdots A(t_{j_k}) \right] \vphantom{\sum_{k=0}^K}\right\| \nn
    &\le  \sum_{k=0}^K \frac 1{k!} \int_{t_\alpha}^{t_\beta}dt_1\int_{t_\alpha}^{t_\beta}dt_2 \cdots \int_{t_\alpha}^{t_\beta}dt_k \, \mathcal{T} \left\| A(t_1) A(t_2) \cdots A(t_k)\right.  \nn
    & \quad  \left. - A(t_{j_1}) A(t_{j_2}) \cdots A(t_{j_k}) \right\| ,
\end{align}
where the convention is taken that the same ordering of $t_{j_\ell}$ is used as for $t_{\ell}$.
For each term for a given $k$, we can upper bound it as
\begin{align}
    &\sum_{\ell=1}^{k} \int_{t_\alpha}^{t_\beta}dt_1\int_{t_\alpha}^{t_\beta}dt_2 \cdots \int_{t_\alpha}^{t_\beta}dt_k \mathcal{T}\left\| A(t_1)\cdots A(t_\ell)A(t_{j_{\ell+1}}) \cdots A(t_{j_k}) \right. \nn
    & \quad \left. - A(t_1)\cdots A(t_{\ell-1})A(t_{j_\ell}) \cdots A(t_{j_k}) \right\| \nn
    & \le \sum_{\ell=1}^{k} \int_{t_\alpha}^{t_\beta}dt_1\int_{t_\alpha}^{t_\beta}dt_2 \cdots \int_{t_\alpha}^{t_\beta}dt_k \mathcal{T}\norm{ A(t_1)\cdots [A(t_\ell)-A(t_{j_\ell})]A(t_{j_{\ell+1}}) \cdots A(t_{j_k})  } \nn
    & = \sum_{\ell=1}^{k} \int_{t_\alpha}^{t_\beta}dt_1\int_{t_\alpha}^{t_\beta}dt_2 \cdots \int_{t_\alpha}^{t_\beta}dt_k \mathcal{T}\norm{ A(t_1)\cdots \left[\int_{t_{j_\ell}}^{t_\ell} ds\, A'(s)\right]A(t_{j_{\ell+1}}) \cdots A(t_{j_k})  } \nn
    & \le A_{\max}^{k-1} \sum_{\ell=1}^{k} \int_{t_\alpha}^{t_\beta}dt_1\int_{t_\alpha}^{t_\beta}dt_2 \cdots \int_{t_\alpha}^{t_\beta}dt_k \int_{t_{j_\ell}}^{t_\ell} ds\, \norm{A'(s)} \nn
    & \le A_{\max}^{k-1} \max_{s\in [t_\alpha,t_\beta]}\norm{A'(s)} \sum_{\ell=1}^{k} \int_{t_\alpha}^{t_\beta}dt_1\int_{t_\alpha}^{t_\beta}dt_2 \cdots \int_{t_\alpha}^{t_\beta}dt_k \int_{t_{j_\ell}}^{t_\ell} ds\,  \nn
    & \le A_{\max}^{k-1} \frac{\delta t\, \Delta t^k}{2} \max_{s\in [t_\alpha,t_\beta]}\norm{A'(s)} .
\end{align}
We can therefore upper bound the error as
\begin{align}
    \norm{W_K(t_\beta,t_\alpha) - \widetilde{W}_K(t_\beta,t_\alpha)} &\le  \sum_{k=1}^K A_{\max}^{k-1} \frac{\delta t \Delta t^k}{k!} \max_{s\in [t_\alpha,t_\beta]}\norm{A'(s)} \nn
    &<  \delta t\, \Delta t \max_{s\in [t_\alpha,t_\beta]}\norm{A'(s)} \sum_{k=1}^\infty \frac{[A_{\max} \Delta t]^k}{k!}  \nn
    &= \delta t\, \Delta t \max_{s\in [t_\alpha,t_\beta]}\norm{A'(s)} \left[ \exp\left( A_{\max} \Delta t\right) -1 \right].
\end{align}
Because we choose $A_{\max} \Delta t\le 1$, the error due to the discretisation can be upper bounded as
\begin{align}
    \norm{W_K(t_\beta,t_\alpha) - \widetilde{W}_K(t_\beta,t_\alpha)} &= \bigO{ \frac{(T/\timesteps)^2}{\smalltime} \max_{s\in [t_\alpha,t_\beta]}\norm{A'(s)} } .
\end{align}
This gives the bound used in Eq.~\eq{Dysdicer}.

Next we bound the error of the particular solution Eq.~\eq{part_sol} when the truncated Dyson series 
\begin{equation}
\bm{v}_K(t_\beta,t_\alpha) = \sum_{k=1}^{K} \frac 1{k!} \int_{t_\alpha}^{t_{\beta}} dt_1 \int_{t_\alpha}^{t_{\beta}} dt_2 \cdots \int_{t_\alpha}^{t_{\beta}} dt_{k} \, \mathcal{T} A(t_1) A(t_2) \cdots A(t_{k-1})\bm{b}(t_{k})\, ,
\end{equation}
is approximated by the discretised integrals
\begin{equation}
\tilde{\bm{v}}_K(t_\beta,t_\alpha) = \sum_{k=2}^{K}\frac{\delta t^k}{k!}\sum_{j_1,\cdots, j_k=0}^{\smalltime-1}\mathcal{T} A(t_{j_1}) \cdots A(t_{j_{k-1}})\bm{b}(t_{j_k}) + \frac {\Delta t}{ \smalltime}\sum_{j=0}^{\smalltime-1}\bm{b}(t_{j}) \, .
\end{equation}
Here the time ordering operator $\mathcal{T}$ is used to mean that the times are sorted in ascending order, \emph{not} that the order of $A$ and $\bm{b}$ is changed, because $\bm{b}$ must always go on the right.

Again we can express the summation in the form of an integral
\begin{equation}
\tilde{\bm{v}}_K(t_\beta,t_\alpha) = \sum_{k=1}^{K} \frac 1{k!} \int_{t_\alpha}^{t_{\beta}} dt_1 \int_{t_\alpha}^{t_{\beta}} dt_2 \cdots \int_{t_\alpha}^{t_{\beta}} dt_{k} \, \mathcal{T} A(t_{j_1}) A(t_{j_2}) \cdots A(t_{j_{k-1}})\bm{b}(t_{k})\, ,
\end{equation}
with the times $t_{j_\ell}$ rounded down to the nearest multiple of $\delta t$.
Now we can upper bound the error in the discretised integrals for each $k$ as
\begin{align}
    &\int_{t_\alpha}^{t_{\beta}} dt_1 \!\int_{t_\alpha}^{t_\beta} dt_2 \cdots \!\int_{t_\alpha}^{t_{\beta}} dt_{k} \, \mathcal{T}\norm{A(t_{1}) A(t_{2}) \cdots A(t_{{k-1}})\bm{b}(t_{j_k})-A(t_{j_1}) A(t_{j_2}) \cdots A(t_{j_{k-1}})\bm{b}(t_{j_k})} \nn
    &\le \sum_{\ell=1}^{k-1} \int_{t_\alpha}^{t_{\beta}} dt_1 \int_{t_\alpha}^{t_\beta} dt_2 \cdots \int_{t_\alpha}^{t_\beta} dt_{k} \, \mathcal{T}\left\| A(t_{1}) \cdots A(t_{{\ell}}) A(t_{j_{\ell+1}}) \cdots A(t_{j_{k-1}})\bm{b}(t_{j_k}) \right. \nn
    & \quad \left. -A(t_{1})  \cdots A(t_{j_{\ell}}) \cdots A(t_{j_{k-1}})\bm{b}(t_{j_k})\right\| \nn
   & \quad + \int_{t_\alpha}^{t_{\beta}} dt_1 \int_{t_\alpha}^{t_\beta} dt_2 \cdots \int_{t_\alpha}^{t_\beta} dt_{k} \, \mathcal{T}\norm{A(t_{1})  \cdots A({t_{k-1}})\bm{b}(t_{k})-A(t_{1}) \cdots A(t_{{k-1}})\bm{b}(t_{j_k})}\nn
   &= \sum_{\ell=1}^{k-1} \int_{t_\alpha}^{t_{\beta}} dt_1 \int_{t_\alpha}^{t_\beta} dt_2 \cdots \int_{t_\alpha}^{t_\beta} dt_{k} \, \mathcal{T}\norm{A(t_{1}) \cdots \int_{t_{j_\ell}}^{t_{\ell}} ds\, A'(s) A(t_{j_{\ell+1}}) \cdots A(t_{{k-1}})\bm{b}(t_{j_k})} \nn
   & \quad + \int_{t_\alpha}^{t_{\beta}} dt_1 \int_{t_\alpha}^{t_\beta} dt_2 \cdots \int_{t_\alpha}^{t_\beta} dt_{k} \, \mathcal{T}\norm{A(t_{1})  \cdots A({t_{k-1}})\int_{t_{j_k}}^{t_{k}}ds\, \bm{b}'(s)} \nn
   & \le A_{\max}^{k-2} b_{\max} \sum_{\ell=1}^{k-1} \int_{t_\alpha}^{t_{\beta}} dt_1 \int_{t_\alpha}^{t_\beta} dt_2 \cdots \int_{t_\alpha}^{t_\beta} dt_{k} \, \int_{t_{j_\ell}}^{t_{\ell}} ds\,\norm{A'(s)}\nn
   & \quad + A_{\max}^{k-1}\int_{t_\alpha}^{t_{\beta}} dt_1 \int_{t_\alpha}^{t_\beta} dt_2 \cdots \int_{t_\alpha}^{t_\beta} dt_{k} \, \int_{t_{j_k}}^{t_{k}}ds\, \norm{\bm{b}'(s)}\nn
   & \le A_{\max}^{k-2} b_{\max} \sum_{\ell=1}^{k-1} \frac{\Delta t^k \delta t}{2} \max_{s\in[t_\alpha,t_\beta]}\norm{A'(s)} 
   + A_{\max}^{k-1} \frac{\Delta t^k \delta t}{2} \max_{s\in[t_\alpha,t_\beta]}\norm{\bm{b}'(s)} .
\end{align}
Now summing over $k$ we have
\begin{align}
   & \norm{{\bm{v}}_K(t_\beta,t_\alpha)-\tilde{\bm{v}}_K(t_\beta,t_\alpha)} \nn
   &< \sum_{k=1}^\infty
    A_{\max}^{k-2} b_{\max} \frac{k-1}{k!} \frac{\Delta t^k \delta t}{2} \max_{s\in[t_\alpha,t_\beta]}\norm{A'(s)}    + \sum_{k=1}^\infty A_{\max}^{k-1} \frac{\Delta t^k \delta t}{2k!} \max_{s\in[t_\alpha,t_\beta]}\norm{\bm{b}'(s)} \nn
   &= \sum_{k=1}^\infty
    A_{\max}^{k-2} b_{\max} \frac{\Delta t^k \delta t}{2(k-1)!} \max_{s\in[t_\alpha,t_\beta]}\norm{A'(s)} 
    -\sum_{k=1}^\infty
    A_{\max}^{k-2} b_{\max} \frac{\Delta t^k \delta t}{2k!} \max_{s\in[t_\alpha,t_\beta]}\norm{A'(s)} \nn
& \quad  
   + \sum_{k=1}^\infty A_{\max}^{k-1} \frac{\Delta t^k \delta t}{2k!} \max_{s\in[t_\alpha,t_\beta]}\norm{\bm{b}'(s)} \nn
   &= \sum_{k=0}^\infty
    A_{\max}^{k-1} b_{\max} \frac{\Delta t^{k+1} \delta t}{2k!} \max_{s\in[t_\alpha,t_\beta]}\norm{A'(s)} 
    - b_{\max} \frac{\delta t}{2A_{\max}^2} \max_{s\in[t_\alpha,t_\beta]}\norm{A'(s)}\left[ \exp\left( A_{\max} \Delta t\right)-1\right]  \nn
    & \quad 
   + \frac{\delta t}{2A_{\max}} \max_{s\in[t_\alpha,t_\beta]}\norm{\bm{b}'(s)}\left[ \exp\left( A_{\max} \Delta t\right)-1\right] \nn
   &= 
    b_{\max} \frac{\Delta t\delta t}{2A_{\max}} \max_{s\in[t_\alpha,t_\beta]}\norm{A'(s)} \exp\left( A_{\max} \Delta t\right)\nn
& \quad    - b_{\max} \frac{\delta t}{2A_{\max}^2} \max_{s\in[t_\alpha,t_\beta]}\norm{A'(s)}\left[ \exp\left( A_{\max} \Delta t\right)-1\right]  \nn
    & \quad 
   + \frac{\delta t}{2A_{\max}} \max_{s\in[t_\alpha,t_\beta]}\norm{\bm{b}'(s)}\left[ \exp\left( A_{\max} \Delta t\right)-1\right] \nn
   &= 
    b_{\max} \frac{\delta t}{2A_{\max}^2} \max_{s\in[t_\alpha,t_\beta]}\norm{A'(s)} \left[ 1- \left(1-A_{\max}\Delta t\right) \exp\left( A_{\max} \Delta t\right)
    \right]\nn
& \quad  
   + \frac{\delta t}{2A_{\max}} \max_{s\in[t_\alpha,t_\beta]}\norm{\bm{b}'(s)}\left[ \exp\left( A_{\max} \Delta t\right)-1\right].
\end{align}
Now with the choice that $A_{\max}\Delta t\le 1$, the upper bound on the error becomes
\begin{equation}\label{eq:erk3plus}
    \bigO{b_{\max} \frac{\Delta t^3}{\smalltime} \max_{s\in[t_\alpha,t_\beta]}\norm{A'(s)} 
   + \frac{\Delta t^2}{\smalltime} \max_{s\in[t_\alpha,t_\beta]}\norm{\bm{b}'(s)}},
\end{equation}
which is equivalent to the upper bound on the error in Eq.~\eq{error_part}.

\end{document}